\documentclass[british, a4paper]{article}

\usepackage{arxiv}

\usepackage{amsmath,amssymb,amsfonts,amsthm}
\usepackage{mathtools}
\usepackage{arydshln}
\usepackage{graphicx}       
\usepackage[utf8]{inputenc} 
\usepackage[T1]{fontenc}    
\usepackage{hyperref}       
\usepackage{url}            
\usepackage{booktabs}       
\usepackage{nicefrac}       
\usepackage{doi}
\usepackage{nth}
\usepackage{float}
\usepackage{caption}
\usepackage{subcaption}
\usepackage[capitalize]{cleveref} 

\usepackage{pgfplots}
\usepackage{tikz}
\usepackage{tikzscale}
\usetikzlibrary{calc}
\newlength\fheight 
\newlength\fwidth 
\pgfplotsset{compat=newest} 
\pgfplotsset{plot coordinates/math parser=false}

\tikzset{sin v source/.style={
		circle,
		draw,
		append after command={
			\pgfextra{
				\draw
				($(\tikzlastnode.center)!0.318!(\tikzlastnode.west)$)
				arc[start angle=180,end angle=0,radius=0.425ex] 
				(\tikzlastnode.center)
				arc[start angle=180,end angle=360,radius=0.425ex]
				($(\tikzlastnode.center)!0.5!(\tikzlastnode.east)$) 
				;
			}
		},
		scale=1.5,
	}
}

\DeclareMathOperator{\vech}{vech}
\DeclareMathOperator{\vect}{vec}
\DeclareMathOperator{\ve}{ve}
\DeclareMathOperator{\cov}{Cov}

\newcommand*{\herm }{\mathsf{H}}
\newcommand{\norm}[1]{\left\lVert#1\right\rVert}

\newtheorem{assum}{Assumption}
\newtheorem{defn}{Definition}
\newtheorem{rmk}{Remark}
\newtheorem{lem}{Lemma}

\graphicspath{ {./pictures/} }

\title{Identification of AC Networks via Online Learning}
\date{\nth{19} September 2021}

\author{ 
    \href{https://orcid.org/0000-0003-3357-4679}{\includegraphics[scale=0.06]{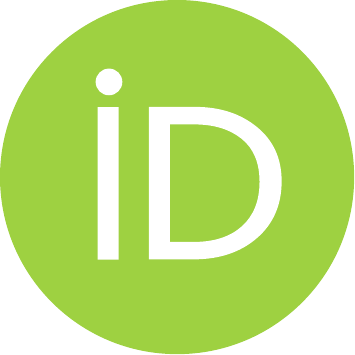}\hspace{1mm}Emanuele Fabbiani} \\
	Identification and Control of Dynamic Systems Laboratory,\\
	University of Pavia,\\
	Pavia, Italy \\
	\texttt{emanuele.fabbiani01@universitadipavia.it} \\
	\And
    \href{https://orcid.org/0000-0002-3505-7731}{\includegraphics[scale=0.06]{orcid.pdf}\hspace{1mm}Pulkit Nahata} \\
	Institute of Mechanical Engineering,\\
	\'Ecole Polytechnique F\'ed\'erale de Lausanne,\\
	Lausanne, Switzerland \\
	\texttt{pulkit.nahata@epfl.ch} \\
	\And
    \href{https://orcid.org/0000-0002-3712-9911}{\includegraphics[scale=0.06]{orcid.pdf}\hspace{1mm}Giuseppe De Nicolao} \\
	Identification and Control of Dynamic Systems Laboratory,\\
	University of Pavia,\\
	Pavia, Italy \\
	\texttt{giuseppe.denicolao@unipv.it} \\
	\And
    \href{https://orcid.org/0000-0002-9492-9624}{\includegraphics[scale=0.06]{orcid.pdf}\hspace{1mm}Giancarlo Ferrari-Trecate} \\
	Institute of Mechanical Engineering,\\
	\'Ecole Polytechnique F\'ed\'erale de Lausanne,\\
	Lausanne, Switzerland \\
	\texttt{giancarlo.ferraritrecate@epfl.ch} \\
}

\renewcommand{\headeright}{Technical Report}
\renewcommand{\undertitle}{Technical Report}
\renewcommand{\shorttitle}{E. Fabbiani et al, Identification of AC Networks via Online Learning}

\hypersetup{
pdftitle={Identification of AC Networks via Online Learning},
pdfauthor={Emanuele Fabbiani, Pulkit Nahata, Giancarlo Ferrari-Trecate, Giuseppe De Nicolao}
}

\begin{document}
\maketitle

\begin{abstract}
	The increasing penetration of intermittent distributed energy resources in power networks calls for novel planning and control methodologies which hinge on detailed knowledge of the grid. However, reliable information concerning the system topology and parameters may be missing or outdated for temporally varying electric distribution networks. This paper proposes an online learning procedure to estimate the network admittance matrix capturing topological information and line parameters. We start off by providing a recursive identification algorithm exploiting phasor measurements of voltages and currents. With the goal of accelerating convergence, we subsequently complement our base algorithm with a design-of-experiment procedure which maximizes the information content of data at each step by computing optimal voltage excitations. Our approach improves on existing techniques, and its effectiveness is substantiated by numerical studies on realistic testbeds.
\end{abstract}


\section{Introduction}
\label{sec:introduction}
Distribution networks serving as an interface between distribution substation and end-to-end customers are going through substantial transformations, attributable to an ever increasing deployment of demand-side technologies and distributed energy resources (DERs). While offering many advantages, DERs can compromise grid reliability due to added intermittency and creation of reverse power flows. In order to ensure safe and resilient operation of distribution systems, comprehensive monitoring and efficient control algorithms are necessary. Nevertheless, any meaningful grid optimization and monitoring task entails grid identification, that is gaining knowledge of grid topology and line parameters.

Research tackling the grid identification problem can broadly be classified into two main branches. On the one hand, works like \cite{bolognani2013identification, deka2018joint, deka2020graphical} propose learning algorithms which draw on the statistical properties of nodal measurements to determine the operational structure and line impedances. This approach has the major advantage of accounting for buses with no available measurements (hidden nodes) \cite{deka2018joint} although restrictive assumptions are required, e.g. hidden nodes must not be adjacent to each other. Moreover, methods based on second-order statistics either make assumptions on the covariance of nodal injections  \cite{bolognani2013identification} or assume its foreknowledge \cite{deka2018joint, deka2020graphical}, and apply only to radial feeders. The latter restriction is dropped in \cite{deka2020graphical}, but only for the purpose of topology estimation. In a realistic setting, these assumptions might not be satisfied; more so due to the rise of distributed generation and smart grids leading to meshed network structures.

On the other hand, in  \cite{yuan2016inverse, babakmehr2016compressive, liao2018urban, ardakanian2019identification}, network identification has been cast into the problem of learning the admittance matrix, where the position of non-zero elements provides topological information, while the values of these are related to the electrical parameters of the lines. Contrary to \cite{bolognani2013identification, deka2018joint, deka2020graphical}, this approach requires voltage, current, or power measurements at each bus of the grid. Nevertheless, it  can be applied to both radial and meshed structures. In particular, Lasso and its variants have been widely adopted to enforce sparsity of the admittance matrix. In \cite{babakmehr2016compressive}, a compressive sensing approach leads to a Lasso formulation to recover the connections of each bus. In \cite{liao2018urban}, a probabilistic graphical model motivates the adoption of Lasso to identify the non-zero elements of the admittance matrix. However, no constraint on the symmetric structure of the admittance matrix is incorporated \textit{a priori}, leading to an over-parameterized solution twice estimating each edge. As a partial remedy to this problem, estimates of the same edge are combined \textit{a posteriori}. While both \cite{babakmehr2016compressive} and \cite{liao2018urban} focus on topology, neither considers the estimation of the electrical parameters of the lines. Finally, in \cite{ardakanian2019identification}, topology and line parameters are obtained at once owing to learning the admittance matrix using Adaptive Lasso. In addition, a procedure to cope with collinearity in measurements is also proposed. 

Different from previously-stated works banking on passively recorded data, an active data collection paradigm is explored in \cite{angjelichinoski2017topology, cavraro2018graph, cavraro2019inverter, du2019optimal}. Grid topology and parameter estimation are complemented with inverter probing in \cite{cavraro2018graph, cavraro2019inverter}. Both works, besides assuming a resistive radial network and employing approximate linearized power-flow equations, lack a comprehensive framework for the optimal design of probing injections. A systematic procedure for maximizing the information content of data samples is explored in \cite{du2019optimal}, wherein active power setpoints for generator nodes are provided by an online design-of-experiment (DoE) procedure \cite{atkinson2007optimum}. Nonetheless, the proposed identification algorithm assumes the availability of line power flows, and neglects the structural constraints of the admittance matrix.
 
All the foregoing works adopt an offline approach, in the sense that they pivot on a batch of previously collected data to estimate grid topology and/or parameters. Distribution networks, unlike transmission networks, oftentimes undergo topological changes for maintenance, load balancing, and fault isolation. Furthermore, future distribution systems are envisaged as reconfigurable networks, wherein certain sections -- just like microgrids -- connect or disconnect to improve dispatch of DERs \cite{shelar2018resilience, Nahata, Dragicevic1}. In the event of a topology change (often localised), a batch method shall discard valuable data, await new samples, and re-run the estimation afresh. On the contrary, an online, recursive identification methodology, encoding the relevant information carried by past data samples in its parameters, can provide new network topology and parameter estimates quickly and autonomously.

\subsection{Paper Contributions and Organization}
This article focuses on AC power networks and introduces an online learning procedure, based only on nodal measurements, for estimating the admittance matrix, which provides detailed information about grid topology and line parameters. The main novelties of this paper are fourfold. First, different from \cite{ardakanian2019identification, yuan2016inverse}, this work proposes a recursive identification algorithm to estimate the admittance matrix, enabling on-the-fly update of topology and fault detection in AC networks that change over time. Second, we provide formulae for deducing a transformation matrix that does away with redundant parameters when the admittance matrix is symmetric and Laplacian. Third, we tap into the principles of optimal experiment design and discuss an approach to compute suitable generator voltages which, when complemented with the base recursive algorithms, accelerates the admittance matrix estimation. Finally, by means of a simulation example, we demonstrate that our method outperforms those existing in literature.   

The remainder of \cref{sec:introduction} introduces relevant preliminaries and notation. \cref{sec:model} recaps network models and motivates the grid identification problem. \cref{sec:recursive} describes the recursive estimation algorithm whereas optimal DoE procedure is discussed in \cref{sec:optimal}. Proposed algorithms are validated via numerical studies in \cref{sec:experiments}. Finally, conclusions are drawn in \cref{sec:conclusions}. 

\subsection{Preliminaries and Notation}
\label{subsec:prem}
\textit{Sets, vectors, matrices, and random variables:} let $j=\sqrt{-1}$ represent the imaginary unit. For a finite set $\mathcal{V}$, $|\mathcal{V}|$ denotes its cardinality. An $(m,n)$ matrix is one with $m$ rows and $n$ columns. Given $x \in \mathbb{C}^{n}$, $\overline{x}$ is its complex conjugate and $[x]$ the associated diagonal matrix of order $n$. Throughout, $1_n$ and $0_n$ are $n$-dimensional vectors of all ones and zeros, whereas $\mathbb{I}_n$ and $\mathbb{O}_{n \times m}$ represent $(n,n)$ identity and $(m,n)$ zero matrices, respectively. The unit vector $e_i,~i= 1,..., n$ is the $i^{th}$ column of $\mathbb{I}_n$. For a matrix $A$, $A^\top$ denotes its transpose, ${A}^\herm $ its Hermitian (complex conjugate) transpose, and $A_{i}$ its $i^{th}$ column vector. The Kronecker product between matrices $A$ and $B$ is $A \otimes B$. A positive definite matrix $A$ and a positive semidefinite matrix $B$ verify $A \succ 0$ and $A \succeq 0$, respectively. We let $\mathcal{N}(x, A)$ designate a Gaussian random vector of dimension $n$, where $x$ is the mean vector and $A$ the covariance matrix.

\textit{Matrix vectorization operators:} 
We indicate by  $\vect(A) =[A_{1\cdot}^\top \cdots A_{n\cdot}^\top]^\top$ the $mn$-dimensional stacked column vector. Furthermore, if $A$ is a square matrix, then the half-vectorization operator $\vech(A)$ provides the $n(n+1)/2$-dimensional vector obtained by eliminating all supradiagonal elements of $A$ from $\vect(A)$. Furthermore, $\ve(A)$ is the $n(n-1)/2$-dimensional vector obtained by removing diagonal elements of $A$ from $-\vech(A)$.  

\section{Network Modeling and Problem Formulation}
\label{sec:model}
In this section, we review relevant algebraic models for AC power networks, and detail the grid identification problem.  

\subsection{Distribution Network Modelling}
\label{subsec:network-model}
An electric distribution network is modeled as an undirected, weighted, and connected graph $\mathcal{G}(\mathcal{V},\mathcal{E},\mathcal{W})$, where the nodes in $\mathcal{V} = \{1, 2, \dots, n\}$ represent buses, either generating units or loads, and edges represent power lines, each connecting two distinct buses and modeled after the standard lumped $\pi-$model \cite{wollenberg1996power}. To each edge $(i, k) \in \mathcal{E}$ we associate a complex weight equal to the line admittance $y_{ik}=g_{ik}+jb_{ik} \in \mathcal{W}$, where $g_{ik} > 0$ is the line conductance and $b_{ik} \in \mathbb{R}$ the line susceptance. The network is then completely represented by the admittance matrix $Y \in \mathbb{C}^{n \times n}$, with elements $Y_{ik}=-y_{ik}$ for $i \neq k$ and $Y_{ik}=\sum_{i=1,i\neq k}^{n}y_{ik} +y_{s,i}$, where $y_{s,i} \in \mathbb{C}$ is the shunt element at the $i^{th}$ bus. If the network does not include phase-shifting transformers and power lines are not compensated by series capacitors, $Y$ is symmetric. In addition, $Y$ is Laplacian if the shunt elements $y_{s,i}$ are not present \cite{Kundur, Taleb}: this happens, for instance, in small- and medium-sized networks, with line lengths less than 60 km.

Throughout this work, we consider a phase-balanced power network operating in sinusoidal regime. To each bus $h \in \mathcal{V}$, we associate a voltage phasor $|v^{(h)}|e^{j\theta^{(h)}}\in \mathbb{C}$, where $|v^{(h)}|$ is the voltage magnitude and $\theta^{(h)} \in \mathbb{R}$ the voltage angle, a current injection phasor $|i^{(h)}|e^{j\phi^{(h)}}\in \mathbb{C}$, and a complex apparent power $s^{(h)}=p^{(h)}+jq^{(h)}$ with $p^{(h)},q^{(h)}\in \mathbb{R}$. As standard in distribution networks, we assume the point of common coupling (PCC) to be the slack bus with fixed $|v^{(0)}|=1$ and $\theta^{(0)}=0$. The remaining buses are classified as generators $\mathcal{S}$ and loads $\mathcal{L}$, such that $\mathcal{V}=\mathcal{S}\cup\mathcal{L}\cup\{0\}$. For notational simplicity we set $|\mathcal{V}|=n$, $|\mathcal{S}|=g$, and $|\mathcal{L}|=l$, where $g,l\geq 1$. In active distribution networks, generators are DERs generally interfaced with inverters equipped with voltage and/or power control \cite{Molzahn}. The current-voltage relation descending directly from Kirchhoff's and Ohm's laws is given by
\begin{equation}
i = Yv,
\label{eq:current-voltage}
\end{equation}
where $i \in \mathbb{C}^{n}$ is the vector of nodal current injections, and $v \in \mathbb{C}^{n}$ the vector of nodal voltages \cite{dorfler2018electrical}. Similarly, one can deduce the relation between the vectors of nodal complex power injections $s$ and nodal voltages $v$ as 
\begin{equation}
{s}=[{v}]{(\overline{Yv})}.
\label{eq:Powerflow}
\end{equation}

\subsection{Identification of AC distribution networks}
The identification problem for AC distribution networks, defined in \cite{yuan2016inverse, ardakanian2019identification}, aims at reconstructing the admittance matrix from a sequence of voltage and current phasor measurements corresponding to different steady states of the system. Similar to \cite{yuan2016inverse, ardakanian2019identification}, our work makes the following assumption.
\begin{assum}
	\label{ass:full-observability}
	The network is fully observable, i.e., voltage and current measurements are available at each node. 
\end{assum}
Let $t$ be the number of samples collected up to a certain time instant, $v_\tau \in \mathbb{C}^n$ and $i_\tau \in \mathbb{C}^n$ the vectors of current injections and voltages for $\tau = 1, \dots, t$. From \eqref{eq:current-voltage}, one can obtain
\begin{equation}
	\label{eq:current-voltage-samples}
	I_t = YV_t,
\end{equation}
where $V_t = [v_1, v_2, \dots, v_t] \in \mathbb{C}^{t \times n}$, and $I_t = [i_1, i_2, \dots, i_t] \in \mathbb{C}^{t \times n}$. The admittance matrix $Y$, encoding both line parameters and topological information, is typically sparse as each bus is not connected to all the remaining nodes. Moreover, as explained in \cref{subsec:network-model}, $Y$ has other structural properties: for most distribution networks, which lack phase-shifting transformers and feature short lines, the following assumption is satisfied.
\begin{assum}
	\label{ass:laplacian-matrix}
	The admittance matrix $Y$ is symmetric and Laplacian, that is, $Y1_n=0_n$. 
\end{assum}
Both the symmetric and Laplacian structures of admittance matrix  greatly reduce the number of entries of Y to be estimated. This observation is further explored in the subsequent section.

\section{Recursive Online Identification \label{sec:recursive}}
In the ideal case of noiseless current and voltage measurements, the identification of $Y$ reduces to solving the system of linear equations \eqref{eq:current-voltage-samples}, once enough samples are collected. Unfortunately, $\mu$PMUs and other metering devices introduce an error commonly modeled as white noise \cite{ardakanian2019identification, deka2018joint}. In the following, for sake of simplicity, it is assumed that the measurement error is distributed as $\mathcal{N}(0_n, \sigma^2\mathbb{I}_n)$, thus implying that the error at each bus has the same variance. As will be clear in the sequel, extensions to different covariance matrices are immediate. 

Upon vectorizing either side of equation \eqref{eq:current-voltage-samples}, one obtains 
\begin{equation}
\vect(I_t) = \vect(YV_t) = \left(V_t^\top \otimes \mathbb{I}_n \right) \vect (Y).
\label{eq:i-v-vectorized}
\end{equation}
Regression methods can be used to get a least squares estimate of $\vect (Y)$ -- the vector representation of the admittance matrix. Before diving into the online estimation algorithm for the admittance matrix $Y$, we note that $\vect (Y)$ comprises $n^2$ entries of $Y$. Being symmetric, $Y$ has at max $n(n+1)/2$ non-redundant entries, which further reduce to $n(n-1)/2$ under Assumption \ref{ass:laplacian-matrix}. Redundant entries in $\vect (Y)$ can be eliminated by means of $\vech (Y)$ -- if $Y$ is symmetric, or $\ve (Y)$ -- when Assumption \ref{ass:laplacian-matrix} holds. Relevant relations between the matrix vectorization operators are summarized in the following Lemma.

\begin{lem}
	\label{lem:duplication}
Given $n$, there is a unique $(n^2, n(n+1)/2)$ matrix $D$, called duplication matrix, such that
\begin{equation}
\vect(Y) = D\vech (Y).
\end{equation}
Furthermore, under Assumption \ref{ass:laplacian-matrix}, there exists a unique $(n(n+1)/2, n(n-1/2))$ matrix $T$ such that 
\begin{equation}
\vech(Y) = T\ve (Y).
\end{equation}
\end{lem}
\begin{proof}
Existence and uniqueness of $D$ are proven in \cite{magnus1980elimination}. The proof of existence and uniqueness of $T$ are in Appendix \ref{app:existence-of-t}, while the construction is shown in Appendix \ref{app:derivation-of-t}. 
\end{proof}
Python and MATLAB codes for constructing $D$ and $T$ are available at \cite{fabbiani2020gists}. Both $T$ and $D$ can be constructed given the number of nodes in the network $n$, therefore they must not be estimated from measurements.

Hereafter, we consider the case where Assumption \ref{ass:laplacian-matrix} holds. Using Lemma \ref{lem:duplication}, we recover the full vectorization of $Y$ as 

\begin{equation}
\vect(Y) = D \vech(Y) = DT\ve (Y).
\label{eq:vectorization}
\end{equation}

By combining \eqref{eq:i-v-vectorized} and \eqref{eq:vectorization} we get
\begin{equation}
\vect(I_t) = \left(V_t^\top \otimes \mathbb{I}_n \right) DT \ve(Y).
\end{equation}
Introducing the following matrices and vectors
\begin{subequations}
    \label{eq:rls-matrices}
	\begin{align}
	A_t &\coloneqq \left(v_t^\top \otimes \mathbb{I}_n\right)DT, \label{eq:rls_a}\\
	\underline{A_t} &\coloneqq \left(V_t^\top \otimes \mathbb{I}_n \right)DT, \\
	b_t &\coloneqq \vect (I_t), ~ \text{and}~ \\
	x &\coloneqq \ve (Y),
	\end{align}
\end{subequations}
the least squares estimation problem at time $t$ writes as
\begin{equation}
\hat x_t = \arg \min_x \norm{b_t - \underline{A_t}x}^2.
\label{eq:least-square}
\end{equation} 
The formulation in \eqref{eq:least-square} equally weights samples at any time instant, which can be detrimental for  time-varying distribution networks and smart grids \cite{ardakanian2017event}. By introducing a forgetting factor $\lambda \in (0, 1]$ \cite{hayes2009statistical}, we reformulate the estimation problem as
\begin{equation}
\hat x_t = \arg \min_x \sum_{i=1}^{t} \lambda^{t-i} \norm{i_i - A_ix }^2. \label{eq:weighted-least-square}
\end{equation}
Given an initial guess of the parameter vector $\hat x_0$ and the matrix $Z_0 \coloneqq \sigma^{-2} \cov[\hat x_0]$, estimates of $\hat x_t$ and $Z_t \coloneqq \sigma^{-2} \cov[\hat x_t]$ can be obtained by the recursive least squares (RLS) algorithm \cite[p. 541]{hayes2009statistical}:
\begin{subequations}
	\label{eq:rls}
	\begin{align}
	\hat x_t &= \hat x_{t-1} + Z_tA_t^\herm  \left(i_t - A_t \hat x_{t-1} \right) \\
	Z_t &= (\lambda Z_{t-1}^{-1} + A_t^\herm  A_t)^{-1} \label{eq:rls_z} \\
	&= \lambda^{-1} (Z_{t-1} - Z_{t-1}A_t^\herm  \left( \lambda \mathbb{I}_n + A_t Z_{t-1} A_t^\herm  \right)^{-1} A_t Z_{t-1}).
	\end{align}
\end{subequations}
From $\hat x_t$, one can derive the estimated admittance matrix $\hat Y_t = DT\hat x_t$. As shown in \cref{sec:model}, the complex elements of the admittance matrix capture both the conductance and the susceptance of the lines. In a real scenario, existing information or batch data can be used to improve the initial guess $x_0$ and $Z_0$. 

The RLS algorithm with constant or bounded forgetting factor is known to have notable stability and convergence properties \cite{liu2013convergence, bittanti1990recursive}. For noisy measurements, RLS with constant forgetting factor is consistent under some excitation conditions only when the forgetting factor is 1 \cite{liu2013convergence}. Otherwise, RLS has limited memory, preventing it from achieving consistency, which is  generally traded off with the ability to follow changes in the parameters. In order to establish a basic degree of competency for the RLS estimator \eqref{eq:rls}, we consider the case of a static network with noise-free measurements. In \cref{sec:experiments} we present numerical simulations to show how the identification method can tolerate noise and can adapt its estimation to changes in network topology.

Classical works establish that, when data are not affected by noise, the error on the parameters is bounded, and its projection onto the subspace for which persistent excitation holds -- see \cite{bittanti1990recursive} for a definition -- converges to zero as the number of samples approaches infinity. Still, the arguments in \cite{bittanti1990recursive} consider only real-valued, single-input-single-output settings. Here, we provide convergence results pertaining to our case, which involves complex inputs, outputs and parameters, and a multivariate output at each iteration.

\begin{lem}
	\label{lem:convergence-of-rls}
	Consider the RLS algorithm \eqref{eq:rls}. Assume that $Y$ is constant in time -- therefore, $x_t = x$, $V_t$ is full-rank and measurement are not affected by noise. Define the error on the parameters $\Tilde{x}_t \coloneqq \hat{x}_t - x$. For any $\hat x_0$ and $Z_0 = Z_0^\herm  \succ 0$, (i) the norm of the error  $\norm{\Tilde{x}_t}$ is bounded, and (ii) the projection of $\Tilde{x}_t$ on the excitation subspace converges to zero as $t$ approaches infinity.
\end{lem}
\begin{proof}
   See Appendix \ref{app:proof-of-rls-convergence}.
\end{proof}

\begin{rmk}
	Recursive least squares assumes that the matrix $V_t$ is full-rank. If not, one can still apply the method to learn part of the admittance matrix, as shown in \cite{ardakanian2019identification}.
\end{rmk}

\begin{rmk}
	RLS algorithm can also be applied to three-phase unbalanced networks. As detailed in \cite{ardakanian2019identification}, the variables to be measured are line-to-ground voltages and current injections for each phase of the nodes, while the admittance matrix to be estimated shares the properties described in \cref{sec:model}.
\end{rmk}

\section{Optimal Design of Experiment}
\label{sec:optimal}
While several learning approaches only capitalize on uncontrolled inputs and outputs, identification algorithms appropriately probing controllable DERs can improve the estimation of the admittance matrix. In this work, each DER is assumed to be equipped with a voltage controller -- necessary for networks with high photovoltaic integration \cite{Molzahn}. Targeting these controllers, we henceforth propose a modified version of the recursive estimation algorithm \eqref{eq:rls} where, at each iteration, DER voltages are set according to a D-optimal design \cite{atkinson2007optimum}, the purpose of which is to maximize the determinant of the Fisher information matrix of the model parameters. With reference to the least squares problem \eqref{eq:weighted-least-square}, the Fisher information matrix \cite{atkinson2007optimum} at time $t$ is $F_t = (\cov (x_t))^{-1}$.
As the measurement noise is a Gaussian vector $\mathcal{N}(\mathbf{0}_n, \sigma^2\mathbb{I}_n)$, we have
\begin{equation}
F_t = \sigma^{-2} Z_t^{-1} = \sigma^{-2} (\lambda Z_{t-1}^{-1} + A_t^\herm A_t).
\end{equation}
We note that $A_t$ depends on the nodal voltages $V_t$; see \eqref{eq:rls_a}. The D-optimal design is the result of the optimization problem
\begin{equation}
v_t^* = \arg \max_{v_t} \det(F_t).
\end{equation}
We observe that $\sigma$ does not influence the optimum and can thus be neglected. Moreover, upon applying the logarithm to the target function -- a common practice for improving numerical properties \cite[Chap. 10]{atkinson2007optimum}, we get
\begin{equation}
v_t^* = \arg \min_{v_t} - \log \det (\lambda Z_{t-1}^{-1} + A_t^\herm A_t).
\end{equation}

While formulating  the DoE problem, we need to take into account voltage limits for all nodes, as well as the active and reactive power dispatched by DERs. Furthermore, the power requirements of loads, expressed by the power flow equations \eqref{eq:Powerflow}, must be satisfied. By adding these constraints, we get the optimization problem
\begin{subequations}
	\label{eq:doe-problem}
	\begin{alignat}{3}
	&(v_t^*, p_t^*) = \arg && \min_{v_t, p_t}  - \log \det (\lambda Z_{t-1}^{-1} A_t^\herm A_t)   \label{eq:doe-target}\\
	&\text{subject to:}      &&  {s_t}=[{v_t}]{(\overline{\hat Y_{t-1}v_t})}  
	&& \label{eq:doe-pf-constraint}\\
	&  &&  v_{\min}^{(i)} \leq v_t^{(i)} \leq v_{\max}^{(i)} 	&&\forall i\in\mathcal{V} \label{eq:doe-v-constraint}\\
	&  &&  \theta_{\min}^{(i)} \leq \theta_t^{(i)} \leq \theta_{\max}^{(i)} &&\forall i\in\mathcal{V} \label{eq:doe-theta-constraint}\\
	&  &&  p_{\min}^{(k)} \leq p_t^{(k)} \leq p_{\max}^{(k)} &&\forall k\in\mathcal{S} \label{eq:doe-p-constraint}\\
	&  &&  q_{\min}^{(k)} \leq q_t^{(k)} \leq q_{\max}^{(k)}  &&\forall k\in\mathcal{S} \label{eq:doe-q-constraint},
	\end{alignat}
\end{subequations}
where $A_t$ depends on $V_t$ as in \eqref{eq:rls_a}.
It is worth noting that the computation of $Z_{t-1}^{-1}$ in \eqref{eq:doe-target} does not require the inversion of $Z_{t-1}$: from \eqref{eq:rls_z}, one has $Z_t^{-1} = \lambda Z_{t-1}^{-1} + A_t^\herm A_t$, which allows for a recursive update of $Z_t^{-1}$.

We also note that constraint \eqref{eq:doe-pf-constraint} depends on $\hat Y_{t-1}$, which is the most recent estimate of the unknown matrix $Y$. While this approximation makes it difficult to analyze the properties of the sequence $\{\hat{Y}_t\}_{t=0}^{\infty}$, numerical experiments described in \cref{sec:experiments} show that such an approach might be only slightly suboptimal with respect to using the real admittance matrix $Y$.

\begin{rmk}
	We note that the proposed DoE procedure helps achieve persistent excitation, which implies the information matrix of the parameters being full rank at each iteration \cite{bittanti1990recursive}. Indeed, since DoE aims at maximizing the determinant of the information matrix, its objective is in contrast with a loss of rank. Therefore, setting voltages as per \eqref{eq:doe-problem} helps satisfying the hypothesis of \cref{lem:convergence-of-rls}.
\end{rmk}
The DoE formulation \eqref{eq:doe-problem} is flexible: one can append more constraints to the optimization problem to cope with technical limitations. For example, the voltage of some DERs may be fixed, or power limitations for certain lines can be introduced. The solution of problem \eqref{eq:doe-problem} is the vector of all nodal voltages; however, voltage references are provided only to DERs as loads cannot generally be controlled.

\begin{rmk}
    The DoE problem \eqref{eq:doe-problem} outputs both voltage magnitude and active power for each generating unit. In this work, we assume that the former is directly used as a control reference, however, the latter can be equivalently adopted in case of power-controlled DERs. When excited with the power reference signal, the generating units cause voltage variations in the network \cite{cavraro2018graph, cavraro2019inverter}: the resulting current-voltage data can then be utilized in \eqref{eq:rls} for the admittance matrix estimation.
\end{rmk}

To summarize, given an initial guess of $\hat{x}_0$ and $\hat{Z}_0$, a value of $\lambda$, and active and reactive power demands for loads, the recursive estimation, enhanced with DoE, can be described by the following steps repeated at each time $t$.
\begin{enumerate}
	\item Solve the DoE problem \eqref{eq:doe-problem} for the nodal voltages $v_t^*$, using the current estimation $\hat{x}_{t-1}$ and $Z_{t-1}$.
	\item Provide the voltage set-point $|v_t^{*, (k)}|$ to DERs $k \in \mathcal{S}$.
	\item Collect measurements of current and voltage phasors from each bus $i \in \mathcal{V}$.
	\item Update the estimates of $\hat{x}_t$ and $Z_t$ using the RLS algorithm \eqref{eq:rls}.
\end{enumerate}

\section{Experiments} \label{sec:experiments}
In order to validate our algorithms, we set up simulations with standard testbeds. As discussed in \cref{sec:introduction}, identification is usually an issue only for distribution networks, while transmission networks are known and constantly monitored. However, to prove the generality of our method, we adopted an example of both a transmission and a distribution network.

\subsection{Experimental Setup}
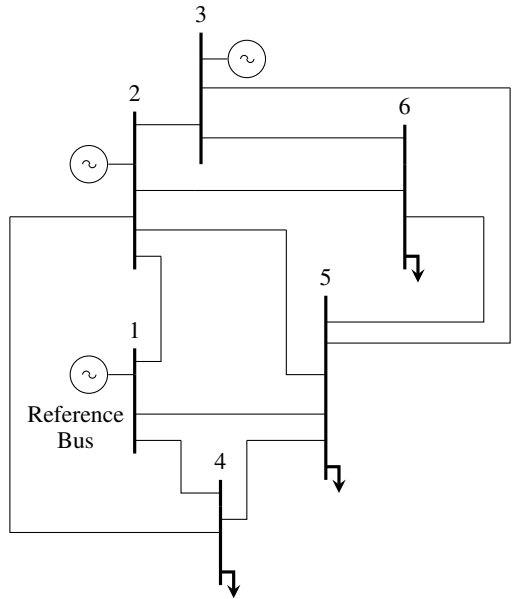
\begin{figure}
	\centering
	\begin{tikzpicture}[scale=0.7]
	
	\node at (-0.25,-0.75) {\small{Reference}};
	\node at (-0.25,-1.25){\small{Bus}};
	
	\draw
	(0,0) node [sin v source] (v1) {}; 
	\draw (v1.east)--++(0.5,0) coordinate(v1-cc); 
	\draw[very thick](v1-cc)--++ (0,0.5) coordinate(v1-up) node[above]{\small{1}}; 
	\draw[very thick](v1-cc)--++(0,-1.5) coordinate(v1-dn); 
	
	
	\draw
	(0,4) node [sin v source] (v2) {}; 
	\draw (v2.east)--++(0.5,0) coordinate(v2-cc); 
	\draw[very thick](v2-cc)--++ (0,1) coordinate(v2-up) node[above]{\small{2}}; 
	\draw[very thick](v2-cc)--++(0,-2) coordinate(v2-dn); 
	
	
	\draw
	(3,6) node [sin v source] (v3) {}; 
	\draw (v3.west)--++(-0.5,0) coordinate(v3-cc); 
	\draw[very thick](v3-cc)--++ (0,0.5) coordinate(v3-up) node[above]{\small{3}}; 
	\draw[very thick](v3-cc)--++(0,-2) coordinate(v3-dn); 
	
	
	\draw coordinate(v4-cc) at (2.5,-2.5); 
	\draw[very thick](v4-cc)--++ (0,0.5) coordinate(v4-up) node[above]{\small{4}}; 
	\draw[very thick](v4-cc)--++(0,-1.5) coordinate(v4-dn); 
	
	
	\draw coordinate(v5-cc) at (4.5,1); 
	\draw[very thick](v5-cc)--++ (0,0.5) coordinate(v5-up) node[above]{\small{5}}; 
	\draw[very thick](v5-cc)--++(0,-3) coordinate(v5-dn); 
	
	
	\draw coordinate(v6-cc) at (6,4); 
	\draw[very thick](v6-cc)--++ (0,.75) coordinate(v6-up) node[above]{\small{6}}; 
	\draw[very thick](v6-cc)--++(0,-2) coordinate(v6-dn); 
	
	\draw ($(v4-cc)-(0,0.5)$)--++(-4,0)--++(0,6) coordinate(i);
	\draw let \p{A}=($(v2-cc)-(i)$) in
	(i) --++ (\x{A},0);
	
	\draw ($(v6-cc)-(0,0.5)$)--++(-4,0) coordinate(i);
	\draw let \p{A}=($(v2-cc)-(i)$) in
	(i) --++ (\x{A},0);
	
	\draw ($(v6-up)-(0,0.25)$) coordinate(i);
	\draw let \p{A}=($(v3-cc)-(i)$) in
	(i) --++ (\x{A},0);
	
	\draw ($(v5-cc)$)--++ (3,0) --++ (0,2) coordinate(i);
	\draw let \p{A}=($(v6-cc)-(i)$) in
	(i) --++ (\x{A},0);
	
	\draw ($(v5-cc)-(0,0.4)$)--++ (3.5,0) --++ (0,4.85) coordinate(i);
	\draw let \p{A}=($(v3-cc)-(i)$) in
	(i) --++ (\x{A},0);
	
	\draw ($(v5-cc)-(0,1.75)$) coordinate(i);
	\draw let \p{A}=($(v1-cc)-(i)$) in
	(i) --++ (\x{A},0);
	
	\draw ($(v5-up)-(0,1.5)$) --++ (-0.75,0)--++(0,2.75) coordinate(i);
	\draw let \p{A}=($(v2-up)-(i)$) in
	(i) --++ (\x{A},0);
	
	\draw ($(v4-up)-(0,0.25)$)--++(-0.75,0)-- ++(0, 1) coordinate(i);
	\draw let \p{A}=($(v1-cc)-(i)$) in
	(i) --++ (\x{A},0);
	
	\draw ($(v1-up)-(0,0.25)$) --++ (0.5,0) --++ (0,2) coordinate(i);
	\draw let \p{A}=($(v2-cc)-(i)$) in
	(i) --++ (\x{A},0);
	
	\draw ($(v2-up)-(0,0.25)$) coordinate(i);
	\draw let \p{A}=($(v3-cc)-(i)$) in
	(i) --++ (\x{A},0);
	
	\draw ($(v4-cc)-(0,0.25)$) --++ (0.5,0) --++ (0,1.5) coordinate(i);
	\draw let \p{A}=($(v5-cc)-(i)$) in
	(i) --++ (\x{A},0);
	
	\draw[very thick,-stealth]($(v4-dn)+(0,0.25)$)--++(0.25,0)--++(0,-0.5cm);
	\draw[very thick,-stealth]($(v5-dn)+(0,0.25)$)--++(0.25,0)--++(0,-0.5cm);
	\draw[very thick,-stealth]($(v6-dn)+(0,0.25)$)--++(0.25,0)--++(0,-0.5cm);
	\end{tikzpicture}
	\caption{A representative diagram of grid T, the 6-bus transmission network \cite{wollenberg1996power}. Buses 1, 2, and 3 are generators, while 4, 5, and 6 are loads.}
	\label{fig:6-bus-grid-diagram}
\end{figure}

\begin{figure}
	\centering
	\includegraphics[scale=0.5]{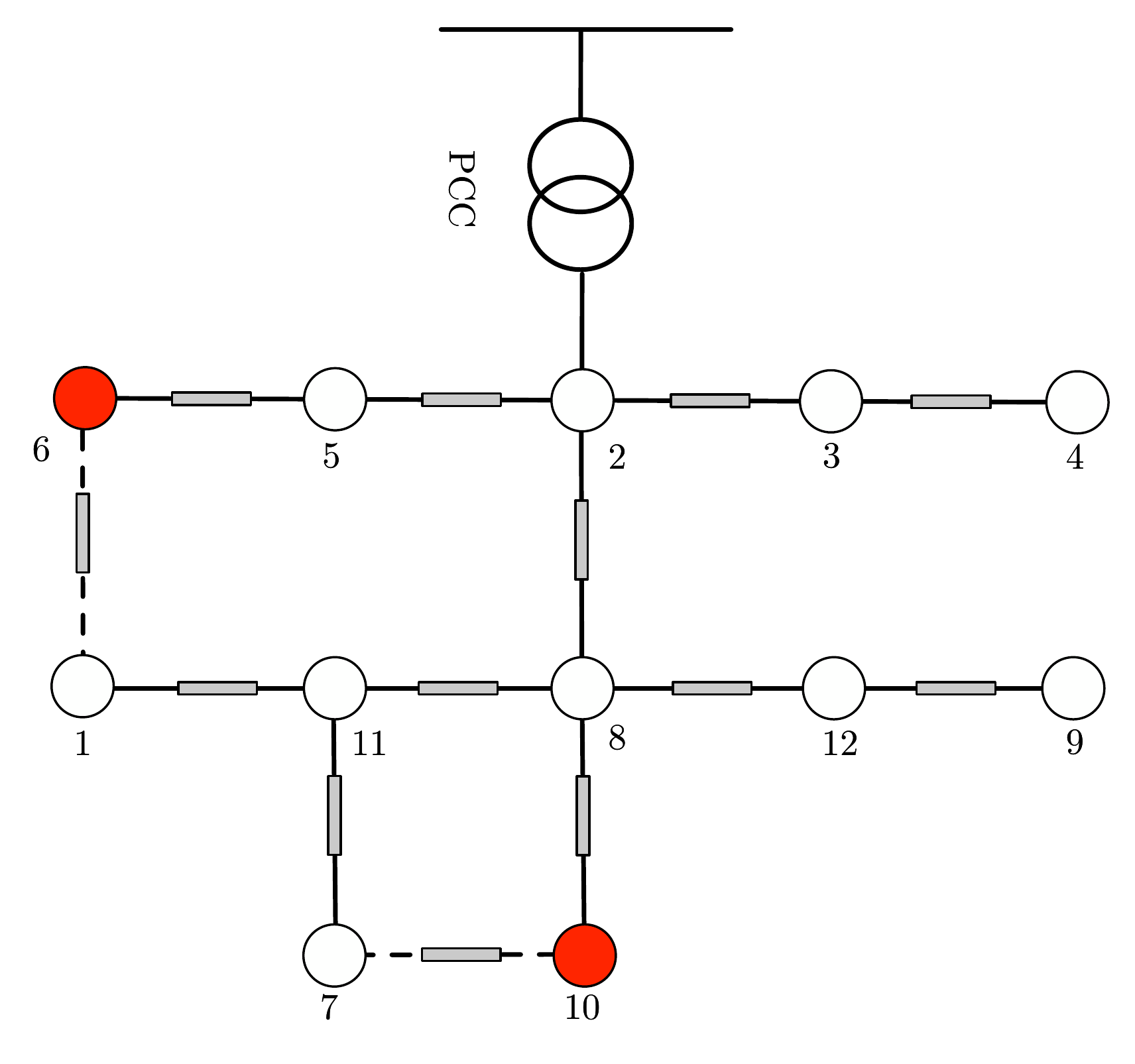}
	\caption{A representative diagram of grid D, the modified IEEE 13-bus feeder \cite{schneider2017analytic}. Buses {\protect\tikz \protect\draw[black, fill=red] (0,0) circle (4pt);} and {\protect\tikz \protect\draw[black] (0,0) circle (4pt);} represent generators and loads, respectively.}
	\label{fig:13-bus-grid-diagram}
	\vspace{-0.5cm}
\end{figure}

We considered two grids: the 6-bus transmission network by Wood and Wollenberg (grid T) \cite[p. 104]{wollenberg1996power} and a modified version of the IEEE 13-bus radial feeder (grid D) \cite{schneider2017analytic}. While the method could scale to much larger networks in theory, in practice collinearity, although mitigated by the design of experiment, still leads to numerical instability for large networks.

In order to test the proposed method on a meshed network, we added two lines to grid D, one connecting bus 1 with 6, and the other bus 7 with 10. As all the lines have negligible capacitance, the admittance matrix of grid D is Laplacian. Conversely, in grid T shunt capacitances are not negligible, resulting in a symmetric, yet non-Laplacian, admittance matrix. 

The presence of controllable generators is a requirement for the application of design of experiment. While grid T features 3 power sources, in grid D distributed generation is introduced through the addition of controllable power sources to buses 6 and 10. Grid D is represented in \cref{fig:13-bus-grid-diagram}, while grid T is displayed in \cref{fig:6-bus-grid-diagram}.

In grid T, load profiles were generated with incorrelated Gaussian active and reactive load fluctuations, centered on the nominal values. This procedure is justified by the observation that, over short periods of time, active and reactive power demands of loads can be modeled as Gaussian random variables \cite{sedghi2015statistical, babakmehr2016compressive}. In grid D, a more realistic setup was adopted: load profiles with one-minute granularity were extracted from the public Pecan Street dataset \cite{pecan2020dataport}. Since this dataset did not include reactive power, a random lagging power factor between 0.85 and 0.95 was considered. Following the procedure adopted in \cite{ardakanian2019identification}, we connected a random number of customers between 5 and 15 to each node. For both grid T and D, we used the AC power flow solver MATPOWER to derive nodal current and voltage phasors \cite{zimmerman2011matpower}.

For each grid, we considered two scenarios to asses the performance of our method in providing an accurate estimate of the admittance matrix. Scenario 1 looks at a network whose topology does not change over time and allows for a comparison between our online algorithm and batch methods like ordinary least squares (OLS) and adaptive lasso \cite{ardakanian2019identification}, whereas scenario 2 considers a time-varying configuration. More specifically, scenario 2 simulates a fault leading to tripping of a line. In grid T, the fault happens on the line connecting bus 2 with 6, while in grid D it impacts the line between bus 7 and 10. Batch algorithms cannot be applied to temporally varying networks and are thus excluded by the tests on scenario 2. In this respect, scenario 2 illustrates the main value of online methods over offline approaches.

We considered three different online estimation methods: 
\begin{itemize}
	\item RLS1, solely imposing the symmetric structure of $Y$ by adopting the parametrization $\boldsymbol{x} = \vech(Y)$;
	\item RLS2, forcing a Laplacian structure of $Y$ by adopting the parametrization $\boldsymbol{x} = \ve(Y)$, as in \eqref{eq:rls-matrices} and \eqref{eq:rls};
	\item DoE, where the generator voltages, excluding the slack bus, are set according to the design-of-experiment procedure presented in \cref{sec:optimal}. The generated inputs and the corresponding outputs are fed to RLS1 if the admittance matrix of the network under consideration if symmetric, and to RLS2 if it is Laplacian.
\end{itemize}
RLS2 was not tested on grid T, as not suitable to the non-Laplacian structure of the admittance matrix of that network. The solution of the design-of-experiment problem \eqref{eq:doe-problem} was computed using an interior-point non-convex solver.

In order to assess the identification performance, we used the error metrics
\begin{subequations}
	\label{eq:power-error-metrics}
	\begin{align}
		m_F & \coloneqq \Vert Y - \hat Y \Vert_\text{F}, \\
		m_{\max} & \coloneqq \Vert Y - \hat Y \Vert_{\max}, \\
		m_R & \coloneqq \Vert Y - \hat Y \Vert_\text{F} / \Vert Y \Vert_\text{F},
	\end{align}
\end{subequations}
where subscripts $F$ and $\max$ denote the Frobenius norm and the max norm, respectively. The metric $m_\text{F}$ assesses the overall goodness of the estimation, $m_{\max}$ is intended to capture possible issues in the identification of single elements, while $m_\text{R}$ provides a relative measure of the identification error.

In all the experiments, we introduced a Gaussian measurement error $\mathcal{N}(\mathbf{0}_n, \sigma^2\mathbb{I}_n)$ on both the real and the imaginary part of the measurements. In both grid T and D, we chose $\sigma$ so that the accuracy $3\sigma$ was 0.1\% of the average magnitude of the measurement, a figure compatible with the characteristics of real metering devices \cite{cavraro2018graph}. This led to $\sigma = 10^{-5}$ in grid D and $\sigma = 10^{-4}$ for grid T. The recursive estimation algorithms were initialized with $\hat{\boldsymbol{x}}_0 = \delta \boldsymbol{1}$, $\delta = 10^{-4}$ and $Z_0 = K \mathbb{I}$, $K = 10^4$, where $\boldsymbol{1}$ and $\mathbb{I}$ have suitable dimensions. The forgetting factor was set to $\lambda = 0.8$.

\subsection{Experimental Results}
For sake of completeness, we present the results on both grids D and T and scenarios 1 and 2. There are little substantial differences, as the following sections show. 

\subsubsection{Grid D}
For scenario 1, \cref{tab:estimation-error-grid-d} shows the comparison with benchmarks after 100 iterations, when the estimates provided by all online algorithms no longer improve. The error metrics can be noticed to be of the same order of magnitude for all methods; although RLS1 and RLS2 achieve poorer performance than OLS and Lasso. This is expected as both OLS and Lasso are batch estimators making use of simultaneous use of all the collected data. We also note that DoE outperforms all other methods, except for Lasso.

\begin{table}
	\centering
	\begin{tabular}{l r r r}
		\hline 
		{} & $m_\text{F} \; [\times 10^{-2}]$ & $m_{\max} \; [\times 10^{-2}]$ & $m_\text{R}$ \\
		\hline 
		OLS (batch) & 5.44 & 1.69 & 0.055\% \\
		Adaptive Lasso (batch) & 2.58 & 0.87 & 0.026\% \\
		RLS1 & 9.55 & 3.84 & 0.095\% \\
		RLS2 & 7.97 & 3.26 & 0.080\% \\
		DoE & 4.74 & 1.27 & 0.047\% \\
		\hline
	\end{tabular}
	\caption{Error metrics grid D, scenario 1, after 100 samples.}
	\label{tab:estimation-error-grid-d}
\end{table}

In both scenarios 1 and 2, DoE achieves faster convergence as well as better accuracy than other iterative methods; see \cref{fig:error-metrics-grid-d}. The downside is the stress on generator voltages, which are subjected to frequent changes (\cref{fig:gen-volt-grid-d}). Nevertheless, due to constraints in the formulation of the design problem \eqref{eq:doe-problem}, both voltage set-points and realized voltages stay within the prescribed interval, which is $[0.95, 1.05]$ p.u. In both scenarios, $m_{\max}$ follows the same trend as $m_\text{F}$ until convergence to a low value, thus ruling out issues about the estimation of specific elements of $Y$.

In the context of scenario 2, the error in the estimation of $y_{7, 10}$ (See \cref{fig:trip-line-estimation-grid-d}) is worth a few comments. Note that $|y_{7, 10}| = 9.8$ up to $t=100$, and subsequently drops to zero as a consequence of the simulated fault. All our recursive implementations are able to quickly adapt to a change in topology, thus proving the usefulness of online estimation. After mere two iterations ($t=102$), the absolute value of the estimated line admittance is 2.21 for RLS1, 2.11 for RLS2, and 1.1 for DoE. Moreover, after 7 iterations, the estimation is lower than 1 for all the online algorithms.

\subsubsection{Grid T}
Results on grid T are aligned with the ones reported for grid D. 

The comparison with benchmarks (\cref{tab:estimation-error}) shows that, after 50 iterations, RLS1 achieves poorer performance than both OLS and Lasso. However, DoE outperforms the batch methods, proving the value of optimal voltage excitations. The difference with grid D may be explained by the higher share of generator in grid T, which enables an higher effectiveness of the design of experiment. The visual comparison between the actual and the estimated admittance matrix shows that all the elements are well estimated. In particular, it is worth noting that the maximum error is 2 orders of magnitudes lower than the smallest element in the admittance matrix. Therefore, inferring the topology of the network from the estimated admittance matrix is trivial. A similar analysis yields the same conclusions on grid D.

\begin{table}
	\centering
	\begin{tabular}{l r r r}
		\hline 
		{} & $m_\text{F} \; [\times 10^{-2}]$ & $m_{\max} \; [\times 10^{-2}]$ & $m_\text{R}$ \\
		\hline 
		OLS (batch) & 3.93 & 1.78 & 0.079\% \\
		Adaptive Lasso (batch) & 3.40 & 1.62 & 0.068\% \\
		RLS1 & 4.84 & 2.41 & 0.097\% \\
		DoE & 1.34 & 0.55 & 0.027\% \\
		\hline
	\end{tabular}
	\caption{Error metrics for grid T, scenario 1 after 50 samples.}
	\label{tab:estimation-error}
\end{table}

DoE achieves faster convergence than RLS1 in both scenarios 1 and 2, as well as better accuracy after 50 iterations - see \cref{fig:error-metrics-grid-t}. The stress posed on generators is comparable to grid D but, coherently with the other test case, voltage set-points and realized voltages never violate the limits, set to $[0.95, 1.05]$ p.u. for bus 2 and $[0.93, 1.07]$ p.u. for bus 3 (\cref{fig:gen-volt-grid-t}).

In the context of scenario 2, it is worth analyzing the error on the estimation of $y_{2,6}$, whose real value becomes zero at time $t=50$ as a consequence of the simulated fault (\cref{fig:trip-line-estimation-grid-t}). After 7 iterations, at $t=57$, the absolute value of the estimation with DoE is 0.48, while it is 2.55 with RLS1. Hence, DoE is again faster in updating the admittance matrix after localized changes.

\subsubsection{Sensitivity to Voltage Noise}
In real applications, measurement noise affects both currents and  voltages. Although a systematic discussion of this scenario is outside the scope of this section, we assess the deterioration in performance experienced by the proposed algorithms when a zero-mean Gaussian noise with covariance matrix $\sigma_v^2\mathbb{I}$ is applied to both the real and the imaginary part of voltage measurements. As displayed in \cref{fig:noise-sensitivity}, all methods suffer from input noise; however, DoE is less affected than other methods, and achieves an acceptable performance even when the noise on voltages is of the same order of magnitude as that on currents.

\subsubsection{Effect of the Design of Experiment Formulation}
As noted in \cref{sec:optimal}, the design-of-experiment formulation \eqref{eq:doe-problem} has to rely on estimated admittance matrix $\hat{Y}_{t-1}$, instead of the unknown real admittance matrix $Y$. In order to show the effect of such an approximation on the identification algorithm, we run DoE on scenario 1 by setting $\hat{Y}_{t-1} = Y$ in \eqref{eq:doe-pf-constraint}. The results in \cref{fig:doe-estimated-actual-y}, produced for grid D, show that the procedure based on the real model of the network performs better; but the difference is marginal. The analysis for grid T yields the same conclusions and it is therefore not reported. 

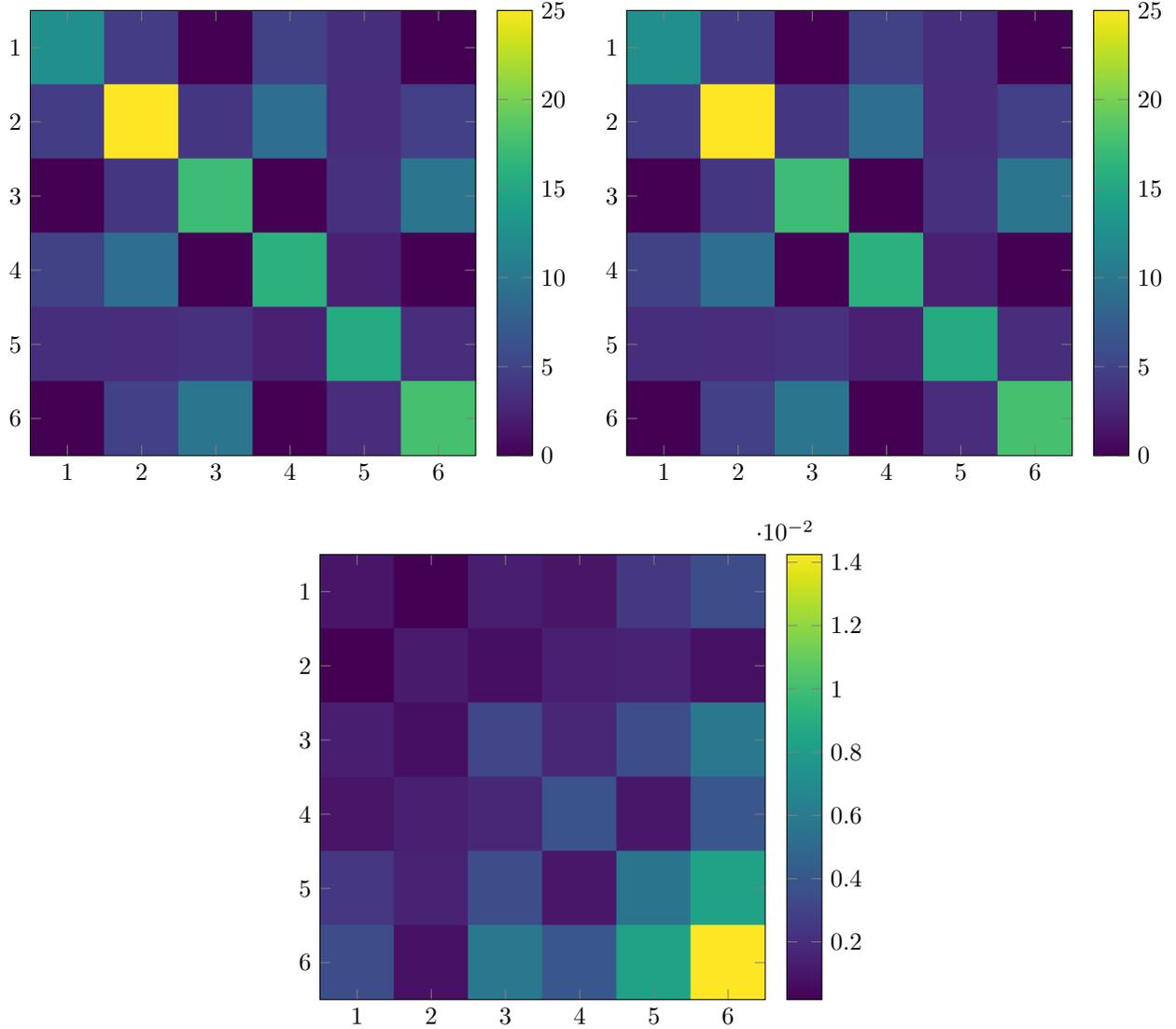
\begin{figure}
	\centering
	\begin{subfigure}[b]{.48\textwidth}
		\begin{tikzpicture}
			\begin{axis}[
				enlargelimits=true,
				height=\textwidth,
				width=\textwidth,
				colormap name=viridis,
				point meta min=0,
				point meta max=25,
				xtick=data,
				ytick=data,
				colorbar
				]
				\addplot[matrix plot,
				mesh/rows=6,
				point meta=explicit,
				] table[x=i, y=j, meta=y, col sep=comma] {"data/estimated_matrices.csv"};
			\end{axis}
		\end{tikzpicture}
	\end{subfigure}
	\hfill
	\begin{subfigure}[b]{.48\textwidth}
		\begin{tikzpicture}
			\begin{axis}[
				enlargelimits=true,
				height=\textwidth,
				width=\textwidth,
				colormap name=viridis,
				point meta min=0,
				point meta max=25,
				xtick=data,
				ytick=data,
				colorbar
				]
				\addplot[matrix plot,
				mesh/rows=6,
				point meta=explicit,
				] table[x=i, y=j, meta=y_doe, col sep=comma] {"data/estimated_matrices.csv"};
			\end{axis}
		\end{tikzpicture}
	\end{subfigure}
	\bigskip 
	\vfill
	\begin{subfigure}[b]{\textwidth}
		\centering
		\begin{tikzpicture}
			\begin{axis}[
				enlargelimits=true,
				height=0.48\textwidth,
				width=0.48\textwidth,
				colormap name=viridis,
				xtick=data,
				ytick=data,
				colorbar
				]
				\addplot[matrix plot,
				mesh/rows=6,
				point meta=explicit,
				] table[x=i, y=j, meta=err, col sep=comma] {"data/estimated_matrices.csv"};
			\end{axis}
		\end{tikzpicture}
	\end{subfigure}
	\caption{Absolute value of the actual and estimated admittance matrix, and of the estimation error after 50 time steps for grid T, scenario 1, using the DoE method. Top left panel: actual admittance matrix $|Y|$, top right panel: estimated admittance matrix $|\hat{Y}_{50}|$, bottom panel: estimation error $|Y - \hat{Y}_{50}|$ -- note the difference in scale with respect to the top panels. \label{fig:heatmap-grid-t}}
\end{figure}

\begin{figure}
	\centering
	\ref*{legend} \\
	\vspace{0.5cm}
	\begin{subfigure}[b]{.48\linewidth}
		\begin{tikzpicture}
		\begin{semilogyaxis}[
		ylabel=$m_\text{F}$,height=6.5cm,width=\linewidth,
		legend columns=-1,
		legend entries={RLS1, RLS2, DoE},
		legend to name=legend,
		ymajorgrids=true,
		xmajorgrids=true,
		grid style=dashed
		]
		\addplot[blue, mark=none] table [y=RLS1,x=SAMPLE, col sep=comma] {"data/fro_errors_without_fault_case_ieee13.csv"};
		\addplot[orange, mark=none] table [y=RLS2,x=SAMPLE, col sep=comma] {"data/fro_errors_without_fault_case_ieee13.csv"};
		\addplot[magenta, mark=none] table [y=DOE,x=SAMPLE, col sep=comma] {"data/fro_errors_without_fault_case_ieee13.csv"};
		\end{semilogyaxis}
		\end{tikzpicture}
		\caption{Frobenius norm of error in scenario 1 \label{fig:fro-err-grid-d-case-1}}
	\end{subfigure}
	\hfill
	\begin{subfigure}[b]{.48\linewidth}
		\begin{tikzpicture}
		\begin{semilogyaxis}[
		height=6.5cm,
		width=\linewidth,
		ymajorgrids=true,
		xmajorgrids=true,
		grid style=dashed
		]
		\addplot[blue, mark=none] table [y=RLS1,x=SAMPLE, col sep=comma] {"data/fro_errors_with_fault_case_ieee13.csv"};
		\addplot[orange, mark=none] table [y=RLS2,x=SAMPLE, col sep=comma] {"data/fro_errors_with_fault_case_ieee13.csv"};
		\addplot[magenta, mark=none] table [y=DOE,x=SAMPLE, col sep=comma] {"data/fro_errors_with_fault_case_ieee13.csv"};
		\end{semilogyaxis}
		\end{tikzpicture}
		\caption{Frobenius norm of error in scenario 2 \label{fig:fro-err-grid-d-case-2}}
	\end{subfigure}   
	\bigskip 
	\vfill
	\begin{subfigure}[b]{.48\linewidth}
		\begin{tikzpicture}
		\begin{semilogyaxis}[
		xlabel=Step,
		ylabel=$m_{\max}$,
		height=6.5cm,
		width=\linewidth,
		ymajorgrids=true,
		xmajorgrids=true,
		grid style=dashed
		]
		\addplot[blue, mark=none] table [y=RLS1,x=SAMPLE, col sep=comma] {"data/max_errors_without_fault_case_ieee13.csv"};
		\addplot[orange, mark=none] table [y=RLS2,x=SAMPLE, col sep=comma] {"data/max_errors_without_fault_case_ieee13.csv"};
		\addplot[magenta, mark=none] table [y=DOE,x=SAMPLE, col sep=comma] {"data/max_errors_without_fault_case_ieee13.csv"};
		\end{semilogyaxis}
		\end{tikzpicture}
		\caption{Max norm of error in scenario 1 \label{fig:max-err-grid-d-case-1}}
	\end{subfigure}
	\hfill
	\begin{subfigure}[b]{.48\linewidth}
		\begin{tikzpicture}
		\begin{semilogyaxis}[
		xlabel=Step,
		height=6.5cm,
		width=\linewidth,
		ymajorgrids=true,
		xmajorgrids=true,
		grid style=dashed
		]
		\addplot[blue, mark=none] table [y=RLS1,x=SAMPLE, col sep=comma] {"data/max_errors_with_fault_case_ieee13.csv"};
		\addplot[orange, mark=none] table [y=RLS2,x=SAMPLE, col sep=comma] {"data/max_errors_with_fault_case_ieee13.csv"};
		\addplot[magenta, mark=none] table [y=DOE,x=SAMPLE, col sep=comma] {"data/max_errors_with_fault_case_ieee13.csv"};
		\end{semilogyaxis}
		\end{tikzpicture}
		\caption{Max norm of error in scenario 2 \label{fig:max-err-grid-d-case-2}}
	\end{subfigure}
	\caption{Error metrics in grid D, scenarios 1 and 2 \label{fig:error-metrics-grid-d}}
\end{figure}
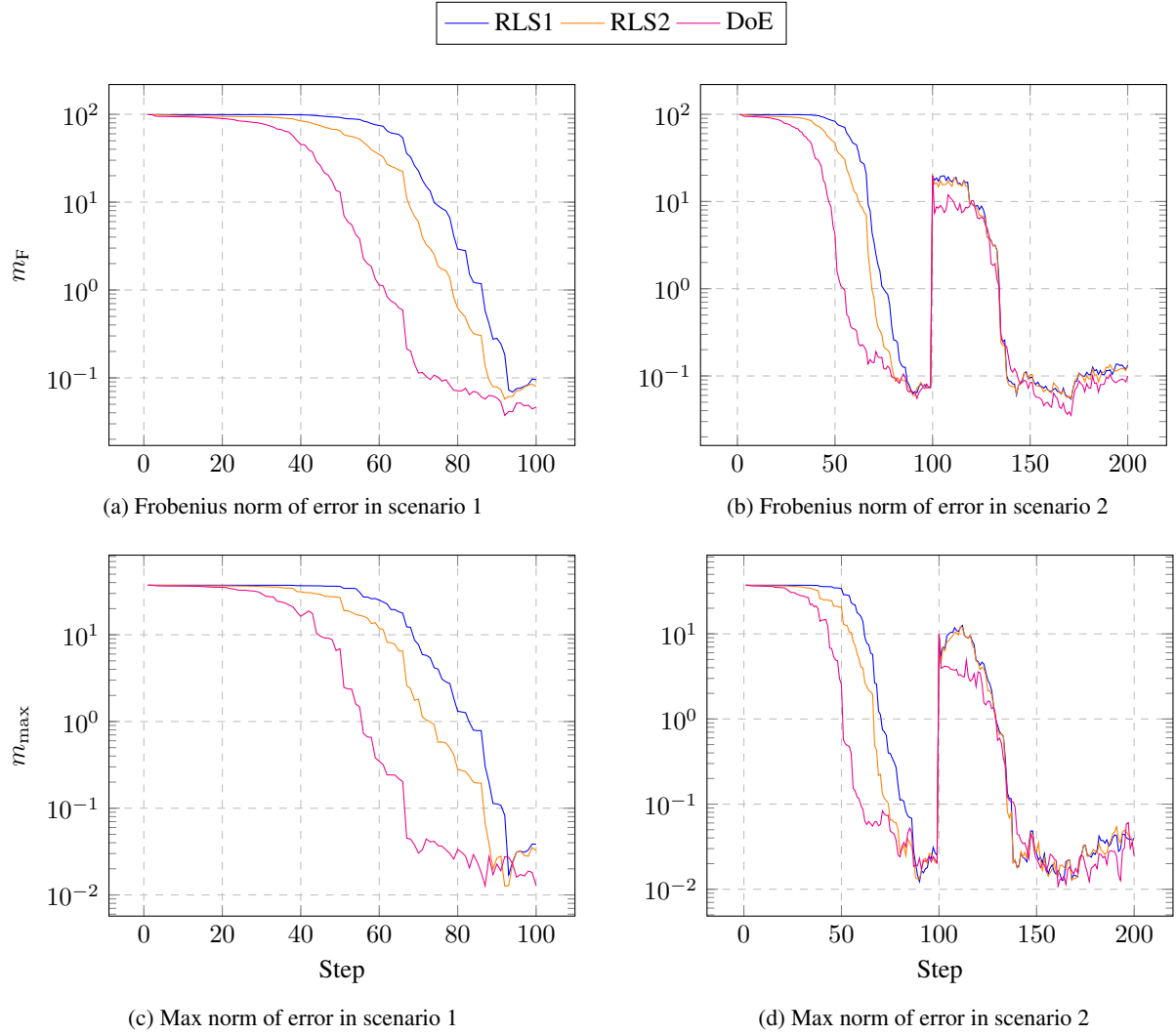

\begin{figure}
	\centering
	\begin{tikzpicture}
	\begin{axis}[
	xlabel=Step,
	ylabel=Bus voltage (p.u.),
	legend pos=north east,
	height=6.5cm,
	width=.9\linewidth,
	legend style={font=\small},
	ymajorgrids=true,
	grid style=dashed
	]
	\addplot[const plot, blue, mark=none] table [y=V1,x=SAMPLE, col sep=comma]{"data/doe_voltages_without_fault_case_ieee13.csv"};
	\addlegendentry{Bus 6}
	\addplot[const plot, magenta, mark=none, dashed] table [y=V2,x=SAMPLE, col sep=comma]{"data/doe_voltages_without_fault_case_ieee13.csv"};
	\addlegendentry{Bus 10}
	\end{axis}
	\end{tikzpicture}
	\caption{Generator voltages produced by DoE for the first 50 iterations in grid D, scenario 1. \label{fig:gen-volt-grid-d}}
\end{figure}
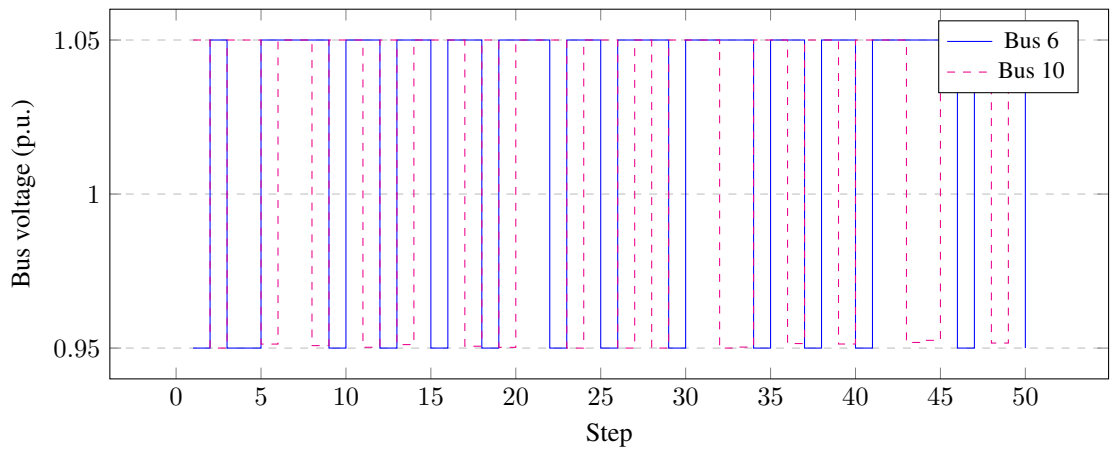

\begin{figure}
	\centering
	\begin{tikzpicture}
	\begin{axis}[
	xlabel=Step,
	ylabel=$\left|y_{7, 10}\right|$,
	legend pos=north east,
	height=6.5cm,
	width=.9\linewidth,
	legend style={font=\small}
	]
	\addplot[blue, mark=none] table [y=RLS1,x=SAMPLE, col sep=comma]{"data/line_trip_estimation_case_ieee13.csv"};
	\addlegendentry{RLS1}
	\addplot[orange, mark=none] table [y=RLS2,x=SAMPLE, col sep=comma]{"data/line_trip_estimation_case_ieee13.csv"};
	\addlegendentry{RLS2}
	\addplot[magenta, mark=none] table [y=DOE,x=SAMPLE, col sep=comma]{"data/line_trip_estimation_case_ieee13.csv"};
	\addlegendentry{DoE}
	\addplot[black, mark=none, dashed] table [y=REAL,x=SAMPLE, col sep=comma]{"data/line_trip_estimation_case_ieee13.csv"};
	\addlegendentry{Real}
	\end{axis}
	\end{tikzpicture}
	\caption{Estimation of element $y_{7, 10}$ in grid D, scenario 2.}
	\label{fig:trip-line-estimation-grid-d}
\end{figure}
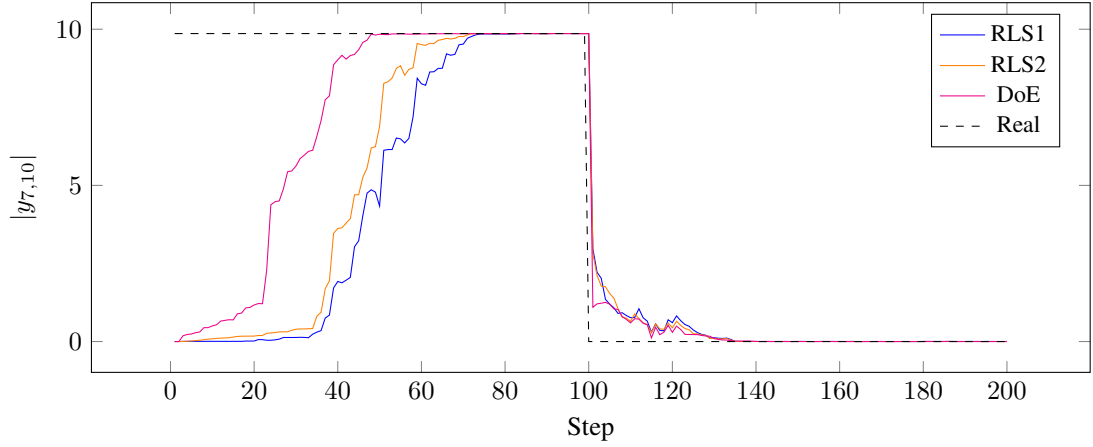

\begin{figure}
	\centering
	\ref*{legend2} \\
	\vspace{0.5cm}
	\begin{subfigure}[b]{.48\linewidth}
		\begin{tikzpicture}
		\begin{semilogyaxis}[
		ylabel=$m_\text{F}$,
		height=6.5cm,
		width=\linewidth,
		legend columns=-1,
		legend entries={RLS1, DoE},
		legend to name=legend2,
		ymajorgrids=true,
		xmajorgrids=true,
		grid style=dashed
		]
		\addplot[blue, mark=none] table [y=RLS1,x=SAMPLE, col sep=comma] {"data/fro_errors_without_fault_case6ww_with_conductance.csv"};
		\addplot[magenta, mark=none] table [y=DOE,x=SAMPLE, col sep=comma] {"data/fro_errors_without_fault_case6ww_with_conductance.csv"};
		\end{semilogyaxis}
		\end{tikzpicture}
		\caption{Frobenius norm of error in scenario 1 \label{fig:fro-err-grid-t-case-1}}
	\end{subfigure}
	\hfill
	\begin{subfigure}[b]{.48\linewidth}
		\begin{tikzpicture}
		\begin{semilogyaxis}[
		height=6.5cm,
		width=\linewidth,
		ymajorgrids=true,
		xmajorgrids=true,
		grid style=dashed
		]
		\addplot[blue, mark=none] table [y=RLS1,x=SAMPLE, col sep=comma] {"data/fro_errors_with_fault_case6ww_with_conductance.csv"};
		\addplot[magenta, mark=none] table [y=DOE,x=SAMPLE, col sep=comma] {"data/fro_errors_with_fault_case6ww_with_conductance.csv"};
		\end{semilogyaxis}
		\end{tikzpicture}
		\caption{Frobenius norm of error in scenario 2 \label{fig:fro-err-grid-t-case-2}}
	\end{subfigure}   
	\bigskip 
	\vfill
	\begin{subfigure}[b]{.48\linewidth}
		\begin{tikzpicture}
		\begin{semilogyaxis}[
		xlabel=Step,
		ylabel=$m_{\max}$,
		height=6.5cm,
		width=\linewidth,
		ymajorgrids=true,
		xmajorgrids=true,
		grid style=dashed
		]
		\addplot[blue, mark=none] table [y=RLS1,x=SAMPLE, col sep=comma] {"data/max_errors_without_fault_case6ww_with_conductance.csv"};
		\addplot[magenta, mark=none] table [y=DOE,x=SAMPLE, col sep=comma] {"data/max_errors_without_fault_case6ww_with_conductance.csv"};
		\end{semilogyaxis}
		\end{tikzpicture}
		\caption{Max norm of error in scenario 1 \label{fig:max-err-grid-t-case-1}}
	\end{subfigure}
	\hfill
	\begin{subfigure}[b]{.48\linewidth}
		\begin{tikzpicture}
		\begin{semilogyaxis}[
		xlabel=Step,
		height=6.5cm,
		width=\linewidth,
		ymajorgrids=true,
		xmajorgrids=true,
		grid style=dashed
		]
		\addplot[blue, mark=none] table [y=RLS1,x=SAMPLE, col sep=comma] {"data/max_errors_with_fault_case6ww_with_conductance.csv"};
		\addplot[magenta, mark=none] table [y=DOE,x=SAMPLE, col sep=comma] {"data/max_errors_with_fault_case6ww_with_conductance.csv"};
		\end{semilogyaxis}
		\end{tikzpicture}
		\caption{Max norm of error in scenario 2 \label{fig:max-err-grid-t-case-2}}
	\end{subfigure}
	\caption{Error metrics in grid T, scenarios 1 and 2 \label{fig:error-metrics-grid-t}}
\end{figure}
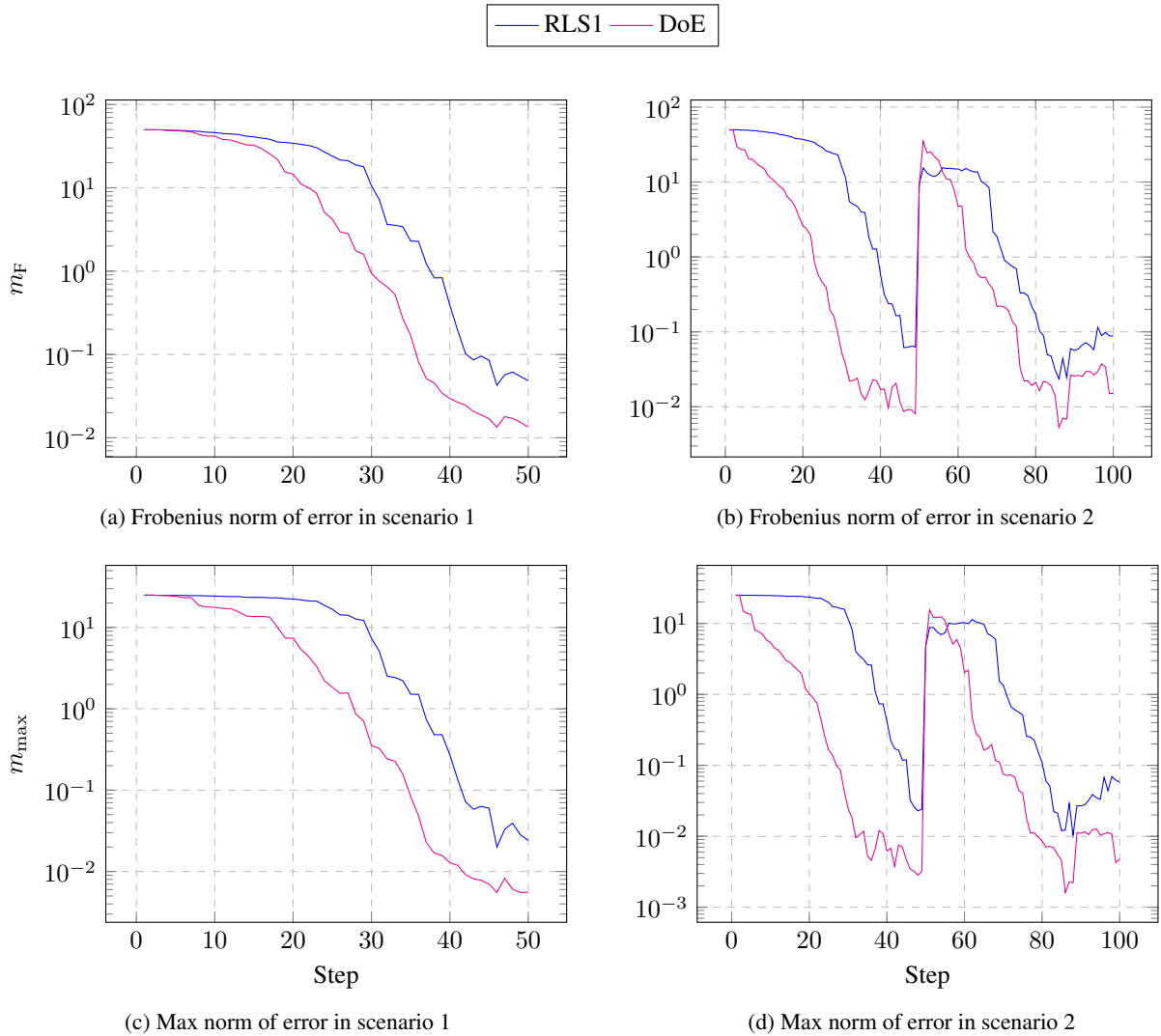

\begin{figure}
	\centering
	\begin{tikzpicture}
	\begin{axis}[
		xlabel=Step,
		ylabel=Bus voltage (p.u.),
		legend pos=north east,
		height=6.5cm,width=\linewidth,
		legend style={font=\small},
		ymajorgrids=true,
		grid style=dashed
	]
	\addplot[const plot, blue, mark=none] table [y=V2,x=SAMPLE, col sep=comma]{"data/doe_voltages_without_fault_case6ww_with_conductance.csv"};
	\addlegendentry{Bus 2}
	\addplot[const plot, magenta, mark=none, dashed] table [y=V3,x=SAMPLE, col sep=comma]{"data/doe_voltages_without_fault_case6ww_with_conductance.csv"};
	\addlegendentry{Bus 3}
	\end{axis}
	\end{tikzpicture}
	\caption{Generator voltages produced by DoE in grid T, scenario 1. \label{fig:gen-volt-grid-t}}
\end{figure}
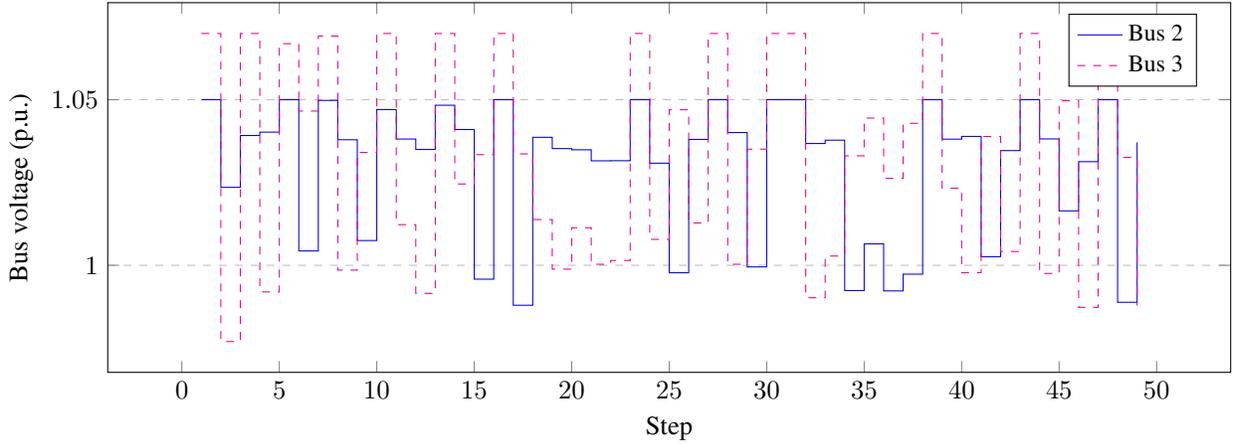

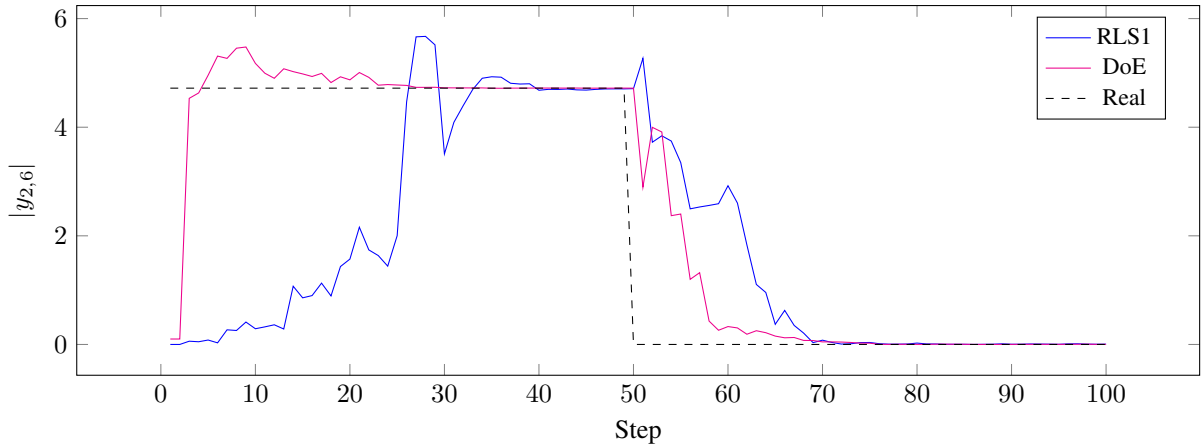
\begin{figure}
	\centering
	\begin{tikzpicture}
	\begin{axis}[
		xlabel=Step,
		ylabel=$\left|y_{2,6}\right|$,
		legend pos=north east,
		height=6.5cm,
		width=\linewidth,
		legend style={font=\small}
	]
	\addplot[blue, mark=none] table [y=RLS1,x=SAMPLE, col sep=comma] {"data/line_trip_estimation_case6ww_with_conductance.csv"};
	\addlegendentry{RLS1}
	\addplot[magenta, mark=none] table [y=DOE,x=SAMPLE, col sep=comma] {"data/line_trip_estimation_case6ww_with_conductance.csv"};
	\addlegendentry{DoE}
	\addplot[black, mark=none, dashed] table [y=REAL,x=SAMPLE, col sep=comma]{"data/line_trip_estimation_case6ww_with_conductance.csv"};
	\addlegendentry{Real}
	\end{axis}
	\end{tikzpicture}
	\caption{Estimation of element $y_{2,6}$ in grid T, scenario 2. \label{fig:trip-line-estimation-grid-t}}
\end{figure}

\begin{figure}
	\centering
	\begin{subfigure}[b]{.48\linewidth}
		\begin{tikzpicture}
		\begin{loglogaxis}[
			xlabel=$\sigma_v$,
			ylabel=$m_\text{F}$,
			legend pos=north west,
			height=6.5cm,
			width=\linewidth,
			legend style={font=\small}
		]
		\addplot[blue, mark=x, dashed, mark options={solid}] table [y=RLS1,x=SD_V, col sep=comma] {"data/noise_sensitivity_fro_errors_case_ieee13.csv"};
		\addlegendentry{RLS1}
		\addplot[orange, mark=+, dashed, mark options={solid}] table [y=RLS2,x=SD_V, col sep=comma]{"data/noise_sensitivity_fro_errors_case_ieee13.csv"};
		\addlegendentry{RLS2}
		\addplot[magenta, mark=star, dashed, mark options={solid}] table [y=DOE,x=SD_V, col sep=comma]{"data/noise_sensitivity_fro_errors_case_ieee13.csv"};
		\addlegendentry{DoE}
		\end{loglogaxis}
		\end{tikzpicture}
	\end{subfigure}
	\hfill
	\begin{subfigure}[b]{.48\linewidth}
		\begin{tikzpicture}
		\begin{loglogaxis}[
			xlabel=$\sigma_v$,
			legend pos=north west,
			height=6.5cm,
			width=\linewidth,
			legend style={font=\small}
		]
		\addplot[blue, mark=x, dashed, mark options={solid}] table [y=RLS1,x=SD_V, col sep=comma] {"data/noise_sensitivity_fro_errors_case6ww.csv"};
		\addlegendentry{RLS1}
		\addplot[magenta, mark=star, dashed, mark options={solid}] table [y=DOE,x=SD_V, col sep=comma]{"data/noise_sensitivity_fro_errors_case6ww.csv"};
		\addlegendentry{DoE}
		\end{loglogaxis}
		\end{tikzpicture}
	\end{subfigure}
	\caption{Frobenius norm of estimation error for different levels of noise on voltage measurements in scenario 1. Left panel: grid D, right panel: grid T. \label{fig:noise-sensitivity}}
\end{figure}
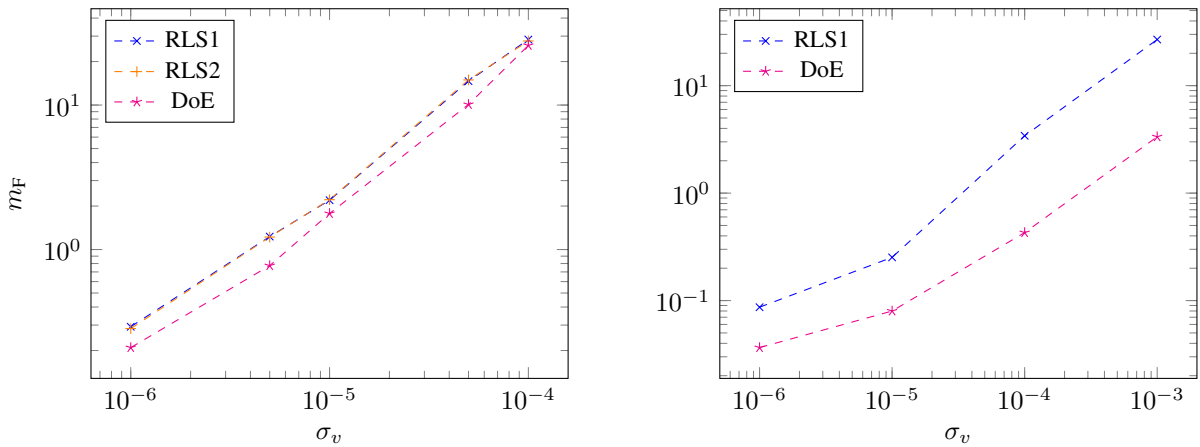

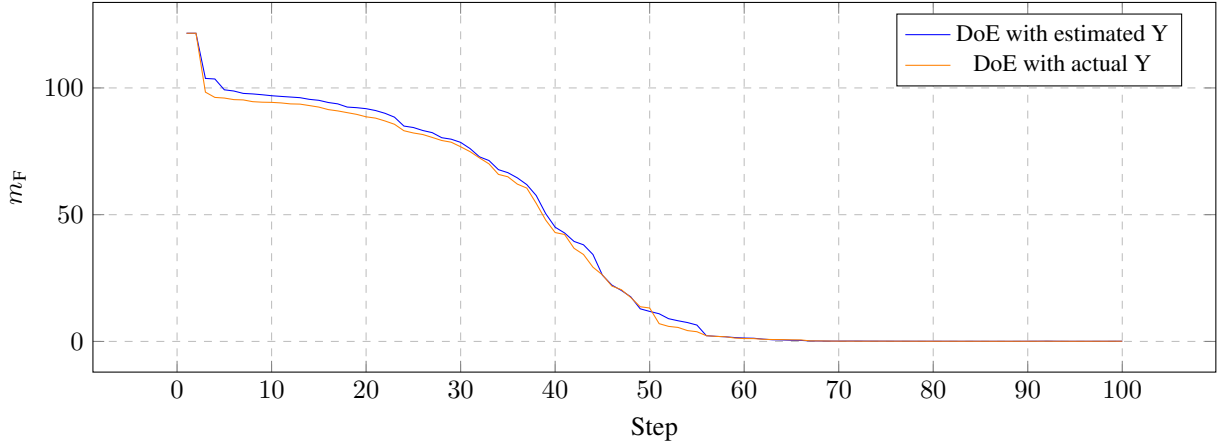
\begin{figure}
	\centering
	\begin{tikzpicture}
	\begin{axis}[
		xlabel=Step,ylabel=$m_\text{F}$,
		legend pos=north east,
		height=6.5cm,
		width=\linewidth,
		legend style={font=\small},
		ymajorgrids=true,
		xmajorgrids=true,
		grid style=dashed
	]
	\addplot[blue, mark=none] table [y=DOE,x=SAMPLE, col sep=comma]{"data/fro_errors_without_fault_case_ieee13_real_Y.csv"};
	\addlegendentry{DoE with estimated Y}
	\addplot[orange, mark=none] table [y=DOE_REAL_Y,x=SAMPLE, col sep=comma]{"data/fro_errors_without_fault_case_ieee13_real_Y.csv"};
	\addlegendentry{DoE with actual Y}
	\end{axis}
	\end{tikzpicture}
	\caption{Comparison between DoE with estimated and actual admittance matrix $Y$. \label{fig:doe-estimated-actual-y}}
\end{figure}

\section{Conclusions}
\label{sec:conclusions}
A frequent lack of detailed information such as grid topology and line parameters motivated the development of this work, which presents an online learning procedure for grid identification in AC networks. In contrast with batch methods for estimating the grid admittance matrix, our algorithm is online and recursive in nature, thus capable of adapting to both sudden changes in the network topology and slow drift in line parameters. 

Notwithstanding the applicability of our methods to generic admittance matrices, we provide a transformation matrix that leverages the structural properties of symmetric Laplacian matrices. Furthermore, we propose a method based on optimal DoE for improving convergence of recursive identification algorithms. 

Future developments will aim at coming up with novel identification techniques for networks where not all nodal electric variables can be measured. Effort will also be devoted to extending out identification framework to error-in-variable models, with a view to properly taking into account all sources of measurement error \cite{brouillon2021bayesian}. Further work will also explore the utility of grid identification schemes in the supervisory control of microgrids \cite{nahata2020hierarchical}.

\appendix
\section{Existence and Uniqueness of the Transformation Matrix $T$}
\label{app:existence-of-t}
We show the existence and uniqueness of the transformation matrix $T$, claimed in Lemma \ref{lem:duplication} introduced in \cref{sec:recursive} and defined as follows.

\begin{defn}
	Given a Laplacian matrix $A \in \mathbb{C}^{n \times n}$, the transformation matrix $T$ is such that 
	\begin{equation}
		\vech(A) = T \ve(A).
	\end{equation}
	\label{def:transformation-matrix}
\end{defn}
\begin{proof}
	Each element of $\vech(A)$ is a linear combination of elements in $\ve(A)$ and this is sufficient to guarantee the existence of a linear map transforming $\ve(A)$ into $\vech(A)$. The uniqueness can be shown by contradiction. Assume there exists $\Tilde{T} \neq T$ such that $\vech(A)=\Tilde{T}\ve(A)=T\ve(A)$. Then, $\mathbf{0}_{n(n+1)/2}=(\Tilde{T} - T)\ve(A), \forall A$. Thus, it has to be $\Tilde{T} = T$. 
\end{proof}

\section{Derivation of the Transformation Matrix $T$}
\label{app:derivation-of-t}
The construction of $T$ is best understood starting with an example. Let $n = 4$ and $A \in \mathbb{C}^{4 \times 4}$ be the following Laplacian matrix:
\begin{equation}
A = \begin{bmatrix}
a_1 + a_2 + a_3 & -a_1 & -a_2 & -a_3\\
-a_1 & a_1 + a_4 + a_5 & -a_4 & -a_5\\
-a_2 & -a_4 & a_2 + a_4 + a_5 & -a_6\\
-a_3 & -a_5 & -a_6 & a_3 + a_5 + a_6\\
\end{bmatrix}
\end{equation} 
By definition, the half-vectorization $\vech(A)$ and the non-redundant vectorization $\ve(A)$ are:
\begin{equation}
\vech(A) = \begin{bmatrix}
a_1 + a_2 + a_3 \\
-a_1 \\
-a_2 \\
-a_3 \\
-a_1 + a_4 + a_5 \\
-a_4 \\ 
-a_5\\
a_2 + a_4 + a_5 \\
-a_6\\
a_3 + a_5 + a_6\\
\end{bmatrix}, \quad
\ve(A) = \begin{bmatrix}
a_1 \\
a_2 \\
a_3 \\
a_4 \\ 
a_5\\
a_6\\
\end{bmatrix}
\end{equation}
From the implicit Definition \ref{def:transformation-matrix}, it is immediate to check that the transformation matrix is:
\begin{equation}
T = \begin{bmatrix}
1 & 1 & 1 & 0 & 0 & 0 \\
-1 & 0 & 0 & 0 & 0 & 0 \\
0 & -1 & 0 & 0 & 0 & 0 \\
0 & 0 & -1 & 0 & 0 & 0 \\
1 & 0 & 0 & 1 & 1 & 0 \\
0 & 0 & 0 & -1 & 0 & 0 \\
0 & 0 & 0 & 0 & -1 & 0 \\
0 & 1 & 0 & 1 & 0 & 1 \\
0 & 0 & 0 & 0 & 0 & -1 \\
0 & 0 & 1 & 0 & 1 & 1 \\
\end{bmatrix}
\end{equation} 
In order to develop a construction procedure for $T$, it is convenient to divide it into $n$ submatrices of different dimensions $T_z$, $z = 1...n$, with $T_z \in \mathbb{R}^{n+1-z \times n(n-1)/2}$ such that:
\begin{equation}
T = \begin{bmatrix}
T_1 \\
\vdots \\
T_z \\ 
\vdots \\
T_n\\
\end{bmatrix}
\end{equation}
Applying the split to the $T$ matrix in the example, we get:
\begin{equation}
T = 
\begin{bmatrix}
T_1 \\
T_2 \\
T_3 \\ 
T_4 \\
\end{bmatrix}
=
\left[\begin{array}{c c c c c c}
1 & 1 & 1 & 0 & 0 & 0 \\
-1 & 0 & 0 & 0 & 0 & 0 \\
0 & -1 & 0 & 0 & 0 & 0 \\
0 & 0 & -1 & 0 & 0 & 0 \\
\hdashline
1 & 0 & 0 & 1 & 1 & 0 \\
0 & 0 & 0 & -1 & 0 & 0 \\
0 & 0 & 0 & 0 & -1 & 0 \\
\hdashline
0 & 1 & 0 & 1 & 0 & 1 \\
0 & 0 & 0 & 0 & 0 & -1 \\
\hdashline
0 & 0 & 1 & 0 & 1 & 1 \\
\end{array}\right]
\end{equation}
Each $T_z$ has a similar structure: when multiplied by $\ve(A)$, the first row yields a diagonal element of $A$, while the other $n-z$ rows adjust the signs of the off-diagonal elements. \\
We thus focus on a generic $T_z$. Its structure can be further divided into four submatrices: the first row is denoted $T_{za}$, while the reminder of $T_i$ can be split into two zero matrices $T_{zb}$ and $T_{zd}$ and a negative identity matrix $T_{zc}$. The sizes of $T_{zb}$, $T_{zc}$, and $T_{zd}$ change with the submatrix index $z$.
We show the split with $T_2$ in the example:
\begin{equation}
T_2 = \left[\begin{array}{c ;{2pt/2pt} c ;{2pt/2pt} c}
\multicolumn{3}{c}{T_{2a}} \\
\hdashline
T_{2b} & T_{2c} & T_{2d}\\
\end{array}\right]
= \left[\begin{array}{c c c ;{2pt/2pt} c c ;{2pt/2pt} c}
1 & 0 & 0 & 1 & 1 & 0 \\
\hdashline
0 & 0 & 0 & -1 & 0 & 0 \\
0 & 0 & 0 & 0 & -1 & 0 \\
\end{array}\right]
\end{equation}
Due to the structure of $\vech(A)$ and $\ve(A)$, one has:
\begin{align}
T_{ib} &= \mathbb{O}_{n-z \times n(n-z) - z(z-1)/2} \\
T_{ic} &= - \mathbb{I}_{n-z} \\
T_{id} &= \mathbb{O}_{n-z \times (n(n-1) + z(z-1))/2 - nz + z}
\end{align}
To justify the expressions, one can observe that every block $T_z$ maps $\ve(A)$ into $n+1-z$ elements of $\vech(A)$, the first being a diagonal element of $A$. The negative identity matrix $T_{zc}$ has size equal to the number of mapped elements of $\vect(A)$ which are not diagonal element of $A$, namely $n+1-z-1=n-z$, while the zero matrix $T_{zb}$ has a number of columns equal to the number of off-diagonal elements of $A$ mapped by $T_w$ with $w<z$, namely $\sum_{k=1}^{z-1}n-k= n(n-z) - z(z-1)/2$. The structure of $T_{zd}$ follows from the size of $T$ and the previous considerations.

The structure of the first row is more complex and reads:
\begin{equation}
T_{za} = \sum_{k=1}^{n-z}\boldsymbol{e}^\top _{k + n(z-1) - z(z-1)/2} + \sum_{k=1}^{z-1}\boldsymbol{e}^\top _{z - 1 + (n-1)(k-1) + k(k-1)/2}
\label{eq:T_i_a}
\end{equation}
where $e \in \mathbb{R}^{n(n-1)/2}$.
The matrices $T_{za}$ map the elements of $\ve(A)$ into the diagonal elements of $A$. By the properties of the Laplacian matrix $A$, its diagonal elements can be expressed as:
\begin{equation}
a_{ii} = - \sum_{j=1, j \neq i}^n a_{ij} = - \left( \sum_{j=1}^{i-1} a_{ij} + \sum_{j=i+1}^{n} a_{ij} \right) = - \left( \sum_{k=1}^{i-1} a_{ki} + \sum_{j=i+1}^{n} a_{ij} \right)
\label{eq:diagonal-elements}
\end{equation}
The first sum in \eqref{eq:T_i_a} accounts for the terms in the second sum in \eqref{eq:diagonal-elements} while the second sum in \eqref{eq:T_i_a} identifies the terms in the first sum of \eqref{eq:diagonal-elements}.

It is worth noting that the construction method described for $T$ is general and holds irrespective of the dimension of $A$. Python and MATLAB implementation of the construction formulae are publicly available on Github \cite{fabbiani2020gists}.

\section{Proof of Lemma 2}
\label{app:proof-of-rls-convergence}
We show here that the results about the convergence of recursive least squares claimed in the work by Bittanti and Bolzern \cite{bittanti1990recursive} also hold in the context presented in \cref{sec:recursive}, featuring complex inputs, outputs, and multivariate output at each iteration.

\begin{proof}
In order to show (i), we start by substituting $\Tilde{\boldsymbol{x}}_t = \hat{\boldsymbol{x}}_t - \boldsymbol{x}$ in \eqref{eq:rls} and considering that, in the noise-free case, $\boldsymbol{i}_t = A_t \boldsymbol{x}$. Then, we get the recursive formula
\begin{equation}
\label{eq:x-tilde-def}
\Tilde{\boldsymbol{x}}_t = \hat{\boldsymbol{x}}_{t-1} - Z_t A_t^\herm  A_t \Tilde{\boldsymbol{x}}_{t-1}.
\end{equation}
For convenience, we define $\boldsymbol{\epsilon}_t \coloneqq A_t \Tilde{\boldsymbol{x}}_{t-1} = A_t \hat{\boldsymbol{x}}_{t-1} - \boldsymbol{i}_t$. Next, we introduce the Lyapunov-like function $W_t \coloneqq \Tilde{x}^\herm  Z_t^{-1}\Tilde{x}_t$. Note that $W$ is a real-valued function, as $Z_t$ is Hermitian and so is $Z_t^{-1}$. By combining the definition of $W_t$ with equations \eqref{eq:x-tilde-def} and \eqref{eq:rls_z}, we derive  
\begin{equation}
\label{eq:vt}
W_t = \lambda W_{t-1} - \boldsymbol{\epsilon}_t^\herm (\mathbb{I}_n - A_tZ_tA_t^\herm )\boldsymbol{\epsilon}_t.
\end{equation}
It can be shown  that $\mathbb{I}_n - A_tZ_tA_t^\herm  \succeq 0$; see e.g. Lemma 1 in \cite{liu2013convergence}. Consider now the quantity $\boldsymbol{\epsilon}_t^\herm (\mathbb{I}_n - A_tZ_tA_t^\herm )\boldsymbol{\epsilon}_t$: it is real and non-negative because $\mathbb{I}_n - A_tZ_tA_t^\herm $ is Hermitian and positive semidefinite. From \eqref{eq:vt}, we obtain the inequality
\begin{equation}
W_t \leq \lambda W_{t-1},
\end{equation}
which one can recursively apply at each $t$ to obtain
\begin{equation}
W_t \leq \lambda^t W_0.
\end{equation}
Recalling the definition of $W$ and \eqref{eq:rls_z}, we write:
\begin{align}
\lambda^t W_0 \geq W_t &= \Tilde{\boldsymbol{x}}_t^\herm  Z_t^{-1}\Tilde{\boldsymbol{x}}_t \\ \nonumber
&= \Tilde{\boldsymbol{x}}_t^\herm \left(\lambda Z_{t-1}^{-1} + A_t^\herm A_t\right)\Tilde{\boldsymbol{x}}_t \\ \nonumber
&= \Tilde{\boldsymbol{x}}_t^\herm \left(\lambda^t Z_{0}^{-1} + \sum_{i=1}^{t} \lambda^{t-i} A_i^\herm A_i\right)\Tilde{\boldsymbol{x}}_t \\ \nonumber
&\geq \lambda^t\Tilde{\boldsymbol{x}}_t^\herm \left(Z_{0}^{-1} + \sum_{i=1}^{t} A_i^\herm A_i\right)\Tilde{\boldsymbol{x}}_t.
\end{align}
Therefore, we conclude that 
\begin{equation}
	\Tilde{\boldsymbol{x}}_t^\herm (Z_{0}^{-1} + \sum_{i=1}^{t} A_i^\herm A_i)\Tilde{\boldsymbol{x}}_t \leq W_0.
\end{equation}
As both $Z_0^{-1} \succ 0$ and $\sum_{i=1}^{t} A_i^\herm A_i \succeq 0$, we have
\begin{equation}
\label{eq:x-tilde-bound}
\Tilde{\boldsymbol{x}}_t^\herm Z_{0}^{-1}\Tilde{\boldsymbol{x}}_t \leq W_0 
\end{equation}
and 
\begin{equation}
\Tilde{\boldsymbol{x}}_t^\herm \left(\sum_{i=1}^{t} A_i^\herm A_i\right)\Tilde{\boldsymbol{x}}_t \leq W_0. \label{eq:x-tilde-bound-with-a}
\end{equation}
Since $\Tilde{\boldsymbol{x}}_t^\herm Z_{0}^{-1}\Tilde{\boldsymbol{x}}_t \geq \text{mineig}(Z_{0}^{-1}) \norm{\Tilde{\boldsymbol{x}}_t}^2$,  equation \eqref{eq:x-tilde-bound} yields
\begin{equation}
	\text{mineig}(Z_0^{-1}) \norm{\Tilde{\boldsymbol{x}}_t}^2 \leq W_0,
\end{equation}
where $\text{mineig}(X)$ is the minimal (real) eigenvalue of a Hermitian matrix $X$. Therefore, $\norm{\Tilde{\boldsymbol{x}}_t}$ is bounded.

In order to show (ii), let $G_t$ be the square root of $\sum_{i=1}^{t} A_i^\herm A_i$, and $\Tilde{\boldsymbol{x}}_t^{(e)}$ and $\Tilde{\boldsymbol{x}}_t^{(u)}$ the projections of $\Tilde{\boldsymbol{x}}_t$ onto the subspaces where persistent excitation holds and does not hold, respectively. Then, \eqref{eq:x-tilde-bound-with-a} can be written as
\begin{equation}
	\norm{G_t\Tilde{\boldsymbol{x}}_t^{(e)} + G_t\Tilde{\boldsymbol{x}}_t^{(u)}} \leq W_0^{1/2}.
\end{equation}
By the reverse triangular inequality, we have:
\begin{equation}
	\label{eq:u-e-projections}
	\underbrace{\norm{\frac{G_t\Tilde{\boldsymbol{x}}_t^{(e)}}{\norm{\Tilde{\boldsymbol{x}}_t^{(e)}}}\norm{\Tilde{\boldsymbol{x}}_t^{(e)}}}}_a - \underbrace{\norm{\frac{G_t\Tilde{\boldsymbol{x}}_t^{(u)}}{\norm{\Tilde{\boldsymbol{x}}_t^{(u)}}}\norm{\Tilde{\boldsymbol{x}}_t^{(u)}}}}_b \leq W_0^{1/2}.
\end{equation}
We can now apply to \eqref{eq:u-e-projections} the same argument proposed in \cite[Proof of Theorem 1]{bittanti1990recursive}. Due to part (i) of the proof, the norms of $\Tilde{\boldsymbol{x}}_t^{(e)}$ and $\Tilde{\boldsymbol{x}}_t^{(u)}$ are bounded. Moreover, as $\Tilde{\boldsymbol{x}}_t^{(u)}$ is the projection of $\Tilde{\boldsymbol{x}}_t$ onto the subspace where persistent excitation does not hold, the term $(b)$ of \eqref{eq:u-e-projections} is bounded. For the inequality \eqref{eq:u-e-projections} to hold, $(a)$ must also be bounded. Yet, the lemma in the appendix of \cite{bittanti1990recursive} shows that $G_t\Tilde{\boldsymbol{x}}_t^{(e)}/\norm{\Tilde{\boldsymbol{x}}_t^{(e)}}$ is unbounded. Therefore, $\Tilde{\boldsymbol{x}}_t^{(e)}$ must converge to zero for \eqref{eq:u-e-projections} to be verified, proving part (ii). It is worth noting that the lemma in \cite{bittanti1990recursive} involves sequences of real positive semidefinite matrices, but the proof holds without modification for complex hermitian positive semidefinite matrices.
\end{proof}

\bibliographystyle{unsrt}
\bibliography{references}

\begin{thebibliography}{10}

\bibitem{bolognani2013identification}
S.~{Bolognani}, N.~{Bof}, D.~{Michelotti}, R.~{Muraro}, and L.~{Schenato}.
\newblock Identification of power distribution network topology via voltage
  correlation analysis.
\newblock In {\em 52nd IEEE Conference on Decision and Control}, pages
  1659--1664, 2013.

\bibitem{deka2018joint}
Deepjyoti Deka, Michael Chertkov, and Scott Backhaus.
\newblock Joint estimation of topology and injection statistics in distribution
  grids with missing nodes.
\newblock {\em arXiv preprint arXiv:1804.04742}, 2018.

\bibitem{deka2020graphical}
D.~{Deka}, S.~{Talukdar}, M.~{Chertkov}, and M.~{Salapaka}.
\newblock Graphical models in meshed distribution grids: Topology estimation,
  change detection \& limitations.
\newblock {\em IEEE Transactions on Smart Grid}, pages 1--1, 2020.

\bibitem{yuan2016inverse}
Ye~Yuan, Omid Ardakanian, Steven Low, and Claire Tomlin.
\newblock On the inverse power flow problem.
\newblock {\em arXiv preprint arXiv:1610.06631}, 2016.

\bibitem{babakmehr2016compressive}
Mohammad Babakmehr, Marcelo~G Sim{\~o}es, Michael~B Wakin, and Farnaz Harirchi.
\newblock Compressive sensing-based topology identification for smart grids.
\newblock {\em IEEE Transactions on Industrial Informatics}, 12(2):532--543,
  2016.

\bibitem{liao2018urban}
Yizheng Liao, Yang Weng, Guangyi Liu, and Ram Rajagopal.
\newblock Urban {MV} and {LV} distribution grid topology estimation via group
  lasso.
\newblock {\em IEEE Transactions on Power Systems}, 34(1):12--27, 2018.

\bibitem{ardakanian2019identification}
Omid Ardakanian, Vincent~WS Wong, Roel Dobbe, Steven~H Low, Alexandra von
  Meier, Claire~J Tomlin, and Ye~Yuan.
\newblock On identification of distribution grids.
\newblock {\em IEEE Transactions on Control of Network Systems}, 6(3):950--960,
  2019.

\bibitem{angjelichinoski2017topology}
M.~{Angjelichinoski}, C.~{Stefanovic}, P.~{Popovski}, A.~{Scaglione}, and
  F.~{Blaabjerg}.
\newblock Topology identification for multiple-bus dc microgrids via primary
  control perturbations.
\newblock In {\em 2017 IEEE Second International Conference on DC Microgrids
  (ICDCM)}, pages 202--206, 2017.

\bibitem{cavraro2018graph}
Guido Cavraro and Vassilis Kekatos.
\newblock Graph algorithms for topology identification using power grid
  probing.
\newblock {\em IEEE control systems letters}, 2(4):689--694, 2018.

\bibitem{cavraro2019inverter}
Guido Cavraro and Vassilis Kekatos.
\newblock Inverter probing for power distribution network topology processing.
\newblock {\em IEEE Transactions on Control of Network Systems}, 6(3):980--992,
  2019.

\bibitem{du2019optimal}
Xu~Du, Alexander Engelmann, Yuning Jiang, Timm Faulwasser, and Boris Houska.
\newblock Optimal experiment design for ac power systems admittance estimation.
\newblock {\em arXiv preprint arXiv:1912.09017}, 2019.

\bibitem{atkinson2007optimum}
Anthony Atkinson, Alexander Donev, Randall Tobias, et~al.
\newblock {\em Optimum experimental designs, with SAS}.
\newblock Oxford University Press, 2007.

\bibitem{shelar2018resilience}
Devendra Shelar, Saurabh Amin, and Ian Hiskens.
\newblock Resilience of electricity distribution networks - part ii: Leveraging
  microgrids, 2018.

\bibitem{Nahata}
Pulkit Nahata, Raffaele Soloperto, Michele Tucci, Andrea Martinelli, and
  Giancarlo Ferrari-Trecate.
\newblock A passivity-based approach to voltage stabilization in {DC}
  microgrids with {ZIP} loads.
\newblock {\em Automatica}, 113:108770, 2020.

\bibitem{Dragicevic1}
T.~Dragi\v{c}evi\'c, X.~Lu, J.~C. Vasquez, and J.~M. Guerrero.
\newblock {DC} microgrids part {I}: A review of control strategies and
  stabilization techniques.
\newblock {\em IEEE Transactions on Power Electronics}, 31(5):4876--4891, 2016.

\bibitem{wollenberg1996power}
Allen~J. Wood and Bruce~F. Wollenberg.
\newblock {\em Power Generation, Operation and Control, 2nd Ed.}
\newblock Wiley-Interscience, 1996.

\bibitem{Kundur}
P.~Kundur.
\newblock {\em {P}ower {S}ystem {S}tability and {C}ontrol}.
\newblock {M}c{G}raw-{H}ill, 1994.

\bibitem{Taleb}
M.~{Taleb}, M.~J. {Ditto}, and T.~{Bouthiba}.
\newblock Performance of short transmission lines models.
\newblock In {\em 2006 IEEE GCC Conference (GCC)}, pages 1--7, 2006.

\bibitem{Molzahn}
D.~K. {Molzahn}, F.~{Dörfler}, H.~{Sandberg}, S.~H. {Low}, S.~{Chakrabarti},
  R.~{Baldick}, and J.~{Lavaei}.
\newblock A survey of distributed optimization and control algorithms for
  electric power systems.
\newblock {\em IEEE Transactions on Smart Grid}, 8(6):2941--2962, 2017.

\bibitem{dorfler2018electrical}
Florian D{\"o}rfler, John~W Simpson-Porco, and Francesco Bullo.
\newblock Electrical networks and algebraic graph theory: Models, properties,
  and applications.
\newblock {\em Proceedings of the IEEE}, 106(5):977--1005, 2018.

\bibitem{magnus1980elimination}
Jan~R Magnus and H~Neudecker.
\newblock The elimination matrix: some lemmas and applications.
\newblock {\em SIAM Journal on Algebraic Discrete Methods}, 1(4):422--449,
  1980.

\bibitem{fabbiani2020gists}
E.~Fabbiani and P.~Nahata.
\newblock Transformation matrix for non-redundant parametrization of laplacian
  matrices.
\newblock Python: \url{https://git.io/JfzHg}, MATLAB:
  \url{https://git.io/JfzH0}, 2020.
\newblock Accessed February 28th, 2020.

\bibitem{ardakanian2017event}
O.~{Ardakanian}, Y.~{Yuan}, R.~{Dobbe}, A.~{von Meier}, S.~{Low}, and
  C.~{Tomlin}.
\newblock Event detection and localization in distribution grids with phasor
  measurement units.
\newblock In {\em 2017 IEEE Power Energy Society General Meeting}, pages 1--5,
  July 2017.

\bibitem{hayes2009statistical}
Monson~H Hayes.
\newblock {\em Statistical digital signal processing and modeling}.
\newblock John Wiley \& Sons, 2009.

\bibitem{liu2013convergence}
Yanjun Liu and Feng Ding.
\newblock Convergence properties of the least squares estimation algorithm for
  multivariable systems.
\newblock {\em Applied Mathematical Modelling}, 37(1-2):476--483, 2013.

\bibitem{bittanti1990recursive}
Sergio Bittanti and Paolo Bolzern.
\newblock Recursive least-squares identification algorithms with incomplete
  excitation: Convergence analysis.
\newblock {\em IEEE Transactions on Automatic Control}, 35(12), 1990.

\bibitem{schneider2017analytic}
KP~Schneider, BA~Mather, BC~Pal, C-W Ten, GJ~Shirek, H~Zhu, JC~Fuller, JLR
  Pereira, LF~Ochoa, LR~De~Araujo, et~al.
\newblock Analytic considerations and design basis for the ieee distribution
  test feeders.
\newblock {\em IEEE Transactions on power systems}, 33(3):3181--3188, 2017.

\bibitem{sedghi2015statistical}
H.~{Sedghi} and E.~{Jonckheere}.
\newblock Statistical structure learning to ensure data integrity in smart
  grid.
\newblock {\em IEEE Transactions on Smart Grid}, 6(4):1924--1933, July 2015.

\bibitem{pecan2020dataport}
Pecan street inc.
\newblock \url{https://www.pecanstreet.org/dataport/}, accessed May 2020.

\bibitem{zimmerman2011matpower}
R.~D. {Zimmerman}, C.~E. {Murillo-SÃ¡nchez}, and R.~J. {Thomas}.
\newblock Matpower: Steady-state operations, planning, and analysis tools for
  power systems research and education.
\newblock {\em IEEE Transactions on Power Systems}, 26(1):12--19, Feb 2011.

\bibitem{brouillon2021bayesian}
Jean-S{\'e}bastien Brouillon, Emanuele Fabbiani, Pulkit Nahata, Florian
  D{\"o}rfler, and Giancarlo Ferrari-Trecate.
\newblock Bayesian error-in-variables models for the identification of power
  networks.
\newblock {\em arXiv preprint arXiv:2107.04480}, 2021.

\bibitem{nahata2020hierarchical}
Pulkit Nahata, Alessio La~Bella, Riccardo Scattolini, and Giancarlo
  Ferrari-Trecate.
\newblock Hierarchical control in islanded dc microgrids with flexible
  structures.
\newblock {\em IEEE Transactions on Control Systems Technology}, 2020.

\end{thebibliography}

\end{document}